\documentclass[a4paper, USenglish, cleveref, autoref, thm-restate, nolineno]{socg-lipics-v2021}

\usepackage{amsmath}
\usepackage{booktabs} 
\usepackage[ruled]{algorithm2e} 
\usepackage{comment}
\usepackage{amsfonts,mathtools,dsfont} 
\usepackage{autonum}
\usepackage[numbers,sort&compress]{natbib}
\usepackage{enumitem}

\Crefname{claim}{Claim}{Claims}

\DeclareRobustCommand{\bm}[1]{\boldsymbol{#1}}

\DeclareSymbolFont{bbold}{U}{bbold}{m}{n}
\DeclareSymbolFontAlphabet{\mathbbold}{bbold}

\DeclareMathSymbol{\mathbboldpi}{\mathord}{bbold}{"19}
\DeclareMathSymbol{\mathbboldtau}{\mathord}{bbold}{"1C}

\newcommand{\bbpi}{\mathbboldpi}

\newcommand{\Prob}[2][\mbox{}]{\mathbf{Pr}_{#1}\left[ #2 \right]}

\newcommand{\Expec}[2][\mbox{}]{\mathbb{E}_{#1}\left[ #2 \right]}
\newcommand{\cExpec}[3][\mbox{}]{\mathbb{E}_{#1}\left[ #2 ~\middle|~ #3 \right]}
\newcommand{\Ind}[1]{\mathds{1}_{\{#1\}}}
\newcommand{\simiid}{\mathrel{\stackrel{\scalebox{0.5}{\text{i.i.d.}}}{\sim}}}

\renewcommand{\epsilon}{\varepsilon}
\newcommand{\bx}{\bm{x}}

\newcommand{\R}{\mathbb{R}}
\newcommand{\E}{\mathbb{E}}
\newcommand{\cE}{\mathcal{E}}

\newcommand{\fD}{\mathbf{D}}
\newcommand{\DD}{\mathfrak{D}}
\newcommand{\bX}{\mathbf{X}}

\newcommand{\cC}{\mathcal{C}}
\newcommand{\cH}{\mathcal{H}}
\newcommand{\cHth}{\tilde{\mathcal{H}}}

\newcommand{\opton}{\mathrm{Opt}_{\rm online}}
\newcommand{\optoff}{\mathrm{Opt}_{\rm offline}}
\newcommand{\regret}{\mathrm{Regret}}

\newcommand{\rade}{\mathrm{Rad}}
\newcommand{\Inj}{\mathrm{Inj}}
\renewcommand{\Pr}{\mathbf{Pr}}
\newcommand{\prwd}{p^{\mathrm{rwd}}}
\newcommand{\pbest}{p^{\mathrm{best}}}
\newcommand{\pls}{p^{\mathrm{ls}}}
\newcommand{\pski}{p^{\mathrm{ski}}}
\newcommand{\history}[1]{\mathbf{s}_{#1 - 1}}
\newcommand{\History}[1]{\mathbf{S}_{#1 - 1}}

\newcommand{\Dadv}{\DD^{\mathrm{adv}}}
\newcommand{\Dro}{\DD^{\mathrm{r.o.}}}
\newcommand{\Diid}{\DD^{\mathrm{i.i.d.}}}
\newcommand{\Dall}{\DD^{\mathrm{all}}}
\newcommand{\Piid}{\Pi^{\mathrm{id}}}

\newcommand{\argmax}{\mathop{\rm arg\,max}}

\DeclareMathOperator{\KL}{KL}
\DeclarePairedDelimiterX{\divbrace}[2]{(}{)}{%
    #1\;\delimsize\|\;#2%
}
\newcommand{\KLdiv}{\KL\divbrace}

\title{Online Algorithms for Repeated Optimal Stopping: \\ Balancing Baseline Guarantees and Regret} 

\titlerunning{Online Algorithms for Repeated Optimal Stopping} 

\author{Tsubasa Harada}{Institute of Science Tokyo, Japan}{harada.t.30af@m.isct.ac.jp}{https://orcid.org/0000-0001-8101-4153}{}
\author{Yasushi Kawase}{Chuo University, Japan}{ykawase422@g.chuo-u.ac.jp}{https://orcid.org/0000-0001-5626-779X}{}
\author{Hanna Sumita}{Institute of Science Tokyo, Japan}{sumita@comp.isct.ac.jp}{https://orcid.org/0000-0003-4005-3206}{}

\authorrunning{T. Harada, Y. Kawase, and H. Sumita} 

\Copyright{Tsubasa Harada, Yasushi Kawase, and Hanna Sumita} 



\category{} 

\relatedversion{} 

\acknowledgements{This work was partially supported by the joint project of Kyoto University and Toyota Motor Corporation, titled ``Advanced Mathematical Science for Mobility Society'', 
JST ERATO Grant Number JPMJER2301,
JST ASPIRE Grant Number JPMJAP2302,
and
JSPS KAKENHI Grant Numbers 
JP21K17708, 
JP21H03397, 
JP25K00137. }

\EventEditors{John Q. Open and Joan R. Access}
\EventNoEds{2}
\EventLongTitle{42nd Conference on Very Important Topics (CVIT 2016)}
\EventShortTitle{CVIT 2016}
\EventAcronym{CVIT}
\EventYear{2016}
\EventDate{December 24--27, 2016}
\EventLocation{Little Whinging, United Kingdom}
\EventLogo{}
\SeriesVolume{42}
\ArticleNo{23}

\begin{document}

\maketitle

\begin{abstract}
We study the \emph{repeated optimal stopping problem}, in which the same optimal stopping instance with an unknown distribution is solved repeatedly over $T$ rounds. 
We aim to simultaneously achieve strong per-round performance guarantees relative to a given baseline and sublinear regret across all rounds.
Our primary contribution is a comprehensive theoretical characterization of whether and when these two objectives are compatible.
First, under standard semi-bandit feedback, we prove that maintaining the per-round guarantee forces regret of $\Omega(T / \log T)$.
Second, even under full feedback, we show that requiring almost-sure satisfaction of the per-round guarantee in every round is incompatible with sublinear regret. 
Third, under full feedback, we propose a general algorithmic framework that achieves both sublinear regret and the per-round guarantee with high probability.
Our framework applies to canonical problems, including the prophet inequality, the secretary problem, and their variants under adversarial, random, and i.i.d.\ input models. 
For example, in the repeated prophet inequality problem, our method guarantees that, with high probability in each round, its expected reward is at least that of the classical single-sample algorithm, which achieves a $1/2$ competitive ratio, while simultaneously ensuring $\tilde{O}(\sqrt{T})$ regret.
Furthermore, we establish a regret lower bound of $\Omega(\sqrt{T})$ even in the i.i.d.\ model,
which is nearly tight with respect to the number of rounds. 



\end{abstract}

\section{Introduction}

Many stopping problems are faced repeatedly under a fixed but unknown distribution. For example, a firm may repeatedly observe a sequence of purchase offers and must decide when to accept an offer to sell. Over repeated rounds, it can adapt its stopping rule using past data. The challenge is to learn a strong policy over time without sacrificing performance in individual rounds.
Such requirements arise naturally in settings where poor decisions in individual rounds are costly, even when long-term performance is satisfactory.

Two classical performance measures capture different aspects of this goal. Regret compares an online strategy to the best online algorithm for the underlying distribution and therefore measures how effectively the algorithm learns over time. However, regret only controls performance over the entire sequence of rounds and does not prevent poor performance in particular rounds, especially during exploration. In contrast, the competitive ratio compares against the offline optimum and provides a robust benchmark for single-round performance. However, a fixed competitive algorithm can be overly conservative and may incur linear regret in repeated settings; see our example in \Cref{sec:example}. This motivates the following perspective: use a baseline algorithm with a provable guarantee against the offline optimum as the benchmark for round-by-round performance, while still seeking low regret over repeated rounds.

We study when strong per-round guarantees relative to such a baseline and sublinear regret can be achieved simultaneously. For repeated optimal stopping, we obtain a nearly complete answer. Under standard semi-bandit feedback, maintaining the per-round guarantee relative to the baseline forces regret of $\Omega(T/\log T)$, which in particular rules out any polynomially sublinear regret bound of the form $O(T^{1-\alpha})$ for any constant $\alpha>0$. Even under full feedback, requiring the per-round guarantee to hold with probability $1$ in every round forces linear regret. Thus, the only remaining setting is full feedback together with a high-probability relaxation of the per-round guarantee, for which we provide a general algorithmic framework with nearly optimal guarantees.

At first glance, the positive result in this setting may appear straightforward.
One might expect that, under full feedback, one can achieve both guarantees by pre-scheduling how often to use the baseline versus a learned policy, for example via explore-then-commit or $\epsilon$-greedy strategies.
However, this intuition is misleading.
As we show in \Cref{sec:hardness}, such non-adaptive mixing strategies fundamentally fail to guarantee both sublinear regret and high-probability per-round performance.
This shows that even in the seemingly favorable full-feedback setting, achieving both guarantees is not a straightforward consequence of standard design patterns.

Taken together, our hardness results identify the unique setting in which the two objectives are compatible: full feedback together with a high-probability relaxation of the baseline constraint. Thus, our positive result is not merely an algorithmic construction, but a conceptual characterization of when baseline-preserving learning is possible at all.
This perspective also explains why the algorithmic component may appear simple: once the feasible regime is tightly characterized, a relatively direct learning-based strategy suffices within it. In this sense, the simplicity of the construction reflects a tight understanding of the problem structure, rather than a lack of technical depth.

Our framework treats the baseline algorithm as a black box. Thus, any sample-based algorithm for the single-round stopping problem can be used as a baseline, and its per-round performance is preserved with high probability while the repeated algorithm simultaneously achieves sublinear regret. This leads to applications to canonical stopping problems including prophet inequality, prophet secretary, secretary-type problems, the last-success problem, and ski-rental.

\subsection{Settings}

Let $\fD=(D_1,\dots,D_n;\Pi)$ be a collection of $n$ independent distributions over $[0,1]$ and a distribution $\Pi$ over permutations of $[n]=\{1,2,\dots,n\}$.
Let $\DD$ be a family of such distributions.
Let $p$ be a known profit function $p\colon [0,1]^n\times [n+1]\to [0,B]$.\footnote{With a slight abuse of terminology, we refer to $p$ as a ``profit function'' even for cost minimization.}
An optimal stopping problem is specified by $(\DD,p)$, and an instance is $(\fD,p)$ with $\fD\in\DD$.

We draw $(\bX, \pi) \sim \fD$, where $\bX=(X_1,\dots,X_n)=(Y_{\pi(1)},\dots,Y_{\pi(n)})$ and $Y_i\sim D_i$.
We write $\bx = (x_1,\ldots, x_n)$ and $\tau$ for realizations of $\bX$ and $\pi$, respectively.
Let $\Inj([i],[n])$ denote injective maps from $[i]$ to $[n]$.
For $i\in[n]$, $\bx\in[0,1]^n$, and $\tau\in\Inj([n],[n])$, we write $\bx_{\le i}\coloneqq(x_1,\dots,x_i)\in[0,1]^i$ and $\tau_{\le i}\coloneqq(\tau(1),\dots,\tau(i))\in\Inj([i],[n])$.

We remark that the distribution $\Pi$ captures the arrival model.
If $\Pi$ is a point-mass at a single permutation, then the problem setting corresponds to the \emph{adversarial order} model, and if $\Pi$ is a uniform distribution, then it corresponds to the \emph{random order} model.
Thus, our framework includes these major arrival models.
In particular, the prophet inequality corresponds to the case where $\Pi$ is a point-mass, while the prophet secretary problem corresponds to the case where $\Pi$ is uniform.

At each step $i\in [n]$, the value $x_i$ and the index $\tau(i)$ are revealed to the algorithm. 
The algorithm knows that $x_i$ is drawn from $D_{\tau(i)}$.
Based on the observed values $\bx_{\le i}$ and $\tau_{\le i}$, the algorithm must decide irrevocably whether to (a) accept $x_i$ and stop (i.e., automatically reject the subsequent values), or (b) reject $x_i$ and proceed to the next step.


Accepting $x_i$ yields a profit of $p(\bx,i)$.
If all values are rejected, the profit is $p(\bx,n+1)$.
The profit function $p$ subsumes the classical prophet-type problems, e.g., the prophet inequality and the prophet secretary problem, in which $p(\bx,i)=x_i$ for $i\in[n]$, and $p(\bx,n+1)=0$, corresponding to the case where no value is accepted.


In the \emph{repeated optimal stopping problem} $(\DD, p; T)$, the same instance $(\fD, p)$ is solved over $T$ rounds.
In each round $t \in [T]$,
a sample $(\bx^{(t)},\tau^{(t)})\sim\fD$ is drawn, and the algorithm makes sequential decisions for $(\bx^{(t)},\tau^{(t)})$.
After each round, feedback is revealed as follows.
In the full-feedback model, the entire pair $(\bx^{(t)},\tau^{(t)})$ is observed.
In the semi-bandit model, letting $i^{(t)}\in[n+1]$ be the stopping time (with $i^{(t)}=n+1$ if no item is accepted), only the prefix $(\bx^{(t)}_{\le i^{(t)}},\tau^{(t)}_{\le i^{(t)}})$ is observed.
Just before round $t$, the algorithm selects a possibly randomized online algorithm $h_t$ for round $t$ based on the feedback observed in the previous $t-1$ rounds.
The goal is to maximize the sum of profits obtained over $T$ rounds.




\subsection{Competitive Ratio and Regret}

Next, we introduce the performance measures of \emph{competitive ratio} and \emph{regret}, highlighting their differences.

A \emph{deterministic online algorithm} for the optimal stopping problem $(\DD,p)$ is defined by a sequence of $n$ functions $h = (f_i)_{i=1}^n$, where each $f_i\colon [0,1]^i\times \Inj([i],[n]) \to \{0,1\}$ takes as input the first $i$ observed values $\bx_{\le i}$ and $\tau_{\le i}$, and outputs $1$ if the algorithm chooses to accept at step $i$, and $0$ otherwise. 
The algorithm proceeds sequentially: it selects the first index $i \in [n]$ such that $f_i(\bx_{\le i},\tau_{\le i}) = 1$, and stops; all subsequent variables are rejected. If $f_i(\bx_{\le i},\tau_{\le i}) = 0$ for all $i \in [n]$, the algorithm rejects all the indices and obtains  $p(\bx,n+1)$.
Let $\cH$ denote the set of all deterministic online algorithms.
A \emph{randomized online algorithm} is a probability distribution over $\cH$.


For a profit function $p$, an algorithm $h = (f_i)_{i\in[n]}$, and a pair $(\bx,\tau)\in[0,1]^n\times \Inj([n],[n])$, we define the profit of $h$ for $(\bx,\tau)$ as 
\begin{align}
    h(\bx,\tau;p) \coloneqq p(\bx, \min \{{i \in [n+1]} \mid i=n+1\text{ or }f_{i}(\bx_{\le i},\tau_{\le i})=1\}).
\end{align}
For an optimal stopping problem instance $(\fD,p)$, the expected profit of $h$ is defined by
\begin{align}
    h(\fD;p) \coloneqq \E_{(\bX,\pi) \sim \fD}[h(\bX,\pi;p)],
\end{align}
where the expectation also includes the internal randomness of $h$ if any.
When $p$ is clear, we write $h(\bx,\tau)$ and $h(\fD)$.
We define the \emph{optimal online profit} by
\begin{align}
\textstyle
\opton(\fD;p)\coloneqq
\sup_{h\in \cH} h(\fD;p)
=\sup_{h\in \cH}\mathbb{E}_{(\bX,\pi)\sim \fD}\left[h(\bX,\pi;p)\right],
\end{align}
and the \emph{optimal offline profit} by
\begin{align}
\textstyle
\optoff(\fD;p)\coloneqq
\mathbb{E}_{(\bX,\pi)\sim \fD}\left[\max_{i\in[n+1]}p(\bX,i)\right],
\end{align}
which corresponds to the expected profit achievable when the realizations are completely known in advance.
For minimization problems, $\sup$ and $\max$ are replaced by $\inf$ and $\min$, respectively. In particular, $\opton(\fD;p)\le \optoff(\fD;p)$ for maximization, and the inequality reverses for minimization.
The \emph{competitive ratio} of $h$ on an instance $(\fD,p)$ is $h(\fD;p)/\optoff(\fD;p)$. For a problem class $(\DD,p)$, it is defined as
\[
\inf_{\fD\in\DD} {h(\fD;p)}/{\optoff(\fD;p)}
\quad
(\text{resp.\ } \sup_{\fD\in\DD} {h(\fD;p)}/{\optoff(\fD;p)})
\]
for maximization (resp.\ minimization).



In the repeated optimal stopping problem $(\fD,p;T)$, the algorithm selects an online algorithm for each round $t$ based on the observation $\history{t}$ in the previous rounds.
In the full-feedback setting, $\history{t}$ consists of full realizations, i.e., $\history{t} = (\bx^{(s)}, \tau^{(s)})_{s \in [t-1]}$, while in the semi-bandit setting it represents the observed prefixes $(\bx^{(s)}_{\le i^{(s)}}, \tau^{(s)}_{\le i^{(s)}})_{s \in [t - 1]}$ up to the stopping times $i^{(s)}$ in each round.
We denote by $\History{t}$ the corresponding random variable.
Let $h_t[\history{t}]$ be the online algorithm chosen in round $t$. 
For a fixed realization $\history{t}$ of past observations, we denote the expected reward in round $t$ by
\[
h_t[\history{t}](\fD; p)
\coloneqq
\mathbb{E}_{(\bX, \pi) \sim \fD}
\left[
h_t[\history{t}](\bX, \pi; p)
\right],
\]
where the expectation is taken over the randomness of the current round and the internal randomness of the algorithm.
Thus, $h_t[\History{t}](\fD; p)$ is a random variable induced by the randomness of past observations.
For brevity, we write $h_t$ when the dependence on $\History{t}$ is implicit, and $h_t(\fD)$ for $h_t(\fD; p)$ when $p$ is clear.
Then, the competitive ratio of the algorithm in round $t$ is defined as
\begin{align}
{\Expec[\History{t}]{h_t[\History{t}](\fD; p)}}\big/{\optoff(\fD;p)}.
\end{align}
The \emph{regret} of the algorithm is defined as
\begin{align}
\regret \coloneqq 
\Expec{\sum_{t=1}^T \left(\opton(\fD;p) - h_t[\History{t}](\fD;p)\right)}.
\label{eq:def-regret}
\end{align}
where the expectation is taken over the randomness of the past observations.
For a cost minimization problem, the regret is defined as the negation of~\eqref{eq:def-regret}. 
Since the regret is trivially at most $B\cdot T$, a regret bound that is sublinear in $T$ is required.



\subsection{Related Work}

Optimal stopping is a central topic in online optimization, with prophet inequality and secretary/best-choice problems as canonical examples.
Their optimal competitive ratios are $1/2$ and $1/e$, respectively~\cite{Krengel1977-uu,dynkin1963optimum}.


When $\fD$ is unknown and only samples are available, one obtains optimal stopping with samples.
For the prophet inequality, a single sample per distribution already yields a $1/2$-competitive policy~\cite{rubinstein_et_al:LIPIcs:2020:11745}; for secretary-type settings see~\cite{DBLP:conf/soda/NutiV23}. See also~\cite{azar2014prophet,correa2020two,correa2024sample,correa2019prophet,correa2021unknown,caramanis2022single,kaplan2020competitive,yoshinaga2024last,cristi2024prophet}.

The \emph{metrical task system} (MTS) has long been a central subject in competitive analysis, and it is also applicable to regret analysis through the framework of \emph{online convex optimization with switching costs}.
As a result, some efforts have been made to design algorithms for MTS that perform well with respect to both competitive ratio and regret~\cite{andrew2013tale,Buchbinder2016,Goel2023,pmlr-v98-daniely19a}.
These studies leverage the repetitive structure inherent in MTS, which makes it challenging to directly apply such approaches to the optimal stopping problem considered in this paper.

Research concerning repeated settings of optimal stopping problems has garnered significant attention in recent years, with particularly vigorous studies focused on the prophet inequality (and the Pandora's box problem) \cite{guo2019settling,fu2020learning,pmlr-v134-guo21a,jin2024sample,Gatmiry2024,agarwal2024semi}. Within the framework of PAC learning,
Jin et al.~\cite{jin2024sample}
showed that $\tilde{O}(1/\epsilon^2)$ samples are sufficient to find an $\epsilon$-additive approximation of the optimal online algorithm with probability at least $1-\delta$ for the prophet inequality.
%
Gatmiry et al.~\cite{Gatmiry2024} proposed an algorithm for the repeated prophet inequality that achieves a regret of $\tilde{O}(n^3 \sqrt{T})$ under partial feedback, where only the values of the selected variables are observed. They also established a regret lower bound of $\Omega(\sqrt{T})$ using two random variables drawn from different distributions.
Agarwal et al.~\cite{agarwal2024semi} introduced a general framework for stochastic optimization problems satisfying monotonicity under semi-bandit feedback, and they derived an upper bound on the regret of $\tilde{O}(n\sqrt{T})$. Recently, this result was improved to $\tilde{O}(\sqrt{nT})$ for the prophet inequality and the Pandora's box problem~\cite{liu2025improved}.
Furthermore, Shah and Rajkumar~\cite{Shah2021} performed a regret analysis for the repeated ski-rental problem.
Note that these studies primarily focus on regret analysis and sample complexity, and do not address the per-round performance.

In contrast, our results show that once a per-round baseline guarantee is imposed, the problem fundamentally differs from standard regret minimization: even in repeated optimal stopping, semi-bandit feedback is insufficient, and full feedback alone does not suffice without relaxing the guarantee to hold with high probability.



\subsection{Hardness Results}
\label{sec:hardness}

One of our main contributions is to provide a comprehensive characterization of the feasibility of achieving both a per-round guarantee against a baseline and sublinear regret.
Our results reveal a fundamental tension between learning for long-term regret minimization and maintaining per-round guarantees. By explicitly ruling out the feasibility of several natural problem settings and approaches, we isolate the exact conditions under which our two objectives can be achieved.
In particular, our results rule out the two most natural directions of extension: 
weakening the feedback model and strengthening the baseline requirement.

\subsubsection{Hardness in the Semi-Bandit Setting}
\label{subsec:hardness-semibandit}

We first consider the semi-bandit setting, which is a natural model for many practical applications. However, we prove that in the semi-bandit setting, it is impossible to simultaneously achieve a high-probability guarantee that the expected profit is at least that of the baseline in each round and a regret of $O(T^{1-\alpha})$ for any $\alpha > 0$.
The proof is deferred to \Cref{app:pf-semibandit-impossible}.

\begin{theorem}
\label{thm:semibandit-impossible}
In the semi-bandit feedback setting,
there exist instances of a repeated optimal stopping problem $(\fD_+,p;T)$ and $(\fD_-,p;T)$, and a sequence of baseline algorithms $(g_t)_{t \in [T]}$ such that the following holds:
for any algorithm sequence $(h_t)_{t \in [T]}$ for $(\fD_+,p;T)$ and $(\fD_-,p;T)$, at least one of the following \emph{does not} hold:
(i)~The regret of $(h_t)_{t \in [T]}$ for $(\fD_+,p;T)$ is $o(T / \log T)$, and
(ii)~$\Prob{h_t(\fD) \ge g_t(\fD)} \ge 1 - O(1/t)$ for each $t\in[T]$ and each $\fD \in \{\fD_+, \fD_-\}$.
\end{theorem}

This theorem implies that our objective cannot be achieved in the semi-bandit setting.
Consequently, we must turn to the full-feedback setting as a necessary structural relaxation.


\subsubsection{Hardness of Almost-Sure Guarantees under Full Feedback}

We now shift our focus to the full-feedback setting, where the algorithm observes the outcomes of all realizations after each round.

One of the difficulties in the semi-bandit setting stems from the per-round guarantee with respect to the baseline. In order to guarantee that the expected profit is at least that of the baseline at each round, the learner may be forced to avoid actions whose profits are uncertain, even if such actions are necessary for learning the underlying distribution. Thus, the baseline requirement fundamentally conflicts with the need to gather informative samples.

In contrast, under full feedback, this difficulty disappears: regardless of the action taken, the learner observes all realizations, and hence can collect information without being constrained by baseline considerations. This might suggest that the problem becomes significantly easier in the full-feedback setting.
In particular, one might hope that it is possible to guarantee that the expected reward is always at least that of the baseline, i.e., with probability one at every round. 

The following theorem shows that this is not the case: even under full feedback, such almost-sure guarantees are incompatible with sublinear regret. See \Cref{sec:impossible} for the proof.

\begin{theorem}
\label{thm:full-feedback-as-impossible}
In the full-feedback setting,
there exist instances of a repeated optimal stopping problem $(\fD_+,p;T)$ and $(\fD_-,p;T)$, and a sequence of baseline algorithms $(g_t)_{t \in [T]}$ such that the following holds:
for any algorithm sequence $(h_t)_{t \in [T]}$ for $(\fD_+,p;T)$ and $(\fD_-,p;T)$, at least one of the following \emph{does not} hold:
(i)~The regret of $(h_t)_{t \in [T]}$ for the instance $(\fD_+,p;T)$ is $o(T)$, and
(ii)~$\Prob{h_t(\fD_-) \ge g_t(\fD_-)} = 1$ for each round $t \in [T]$.
\end{theorem}

This strict impossibility dictates that the safety guarantee must be relaxed to hold with high probability.

\subsubsection{Hardness of High-Probability Guarantees under Full Feedback}

The preceding results rule out two natural approaches for achieving a guarantee with respect to a baseline.
In the semi-bandit setting, \Cref{thm:semibandit-impossible} shows that even high-probability guarantees are incompatible with sublinear regret.
In the full-feedback setting, \Cref{thm:full-feedback-as-impossible} further shows that requiring almost-sure guarantees is also impossible.
Taken together, these results leave only one remaining possibility: to achieve a per-round guarantee with high probability in the full-feedback setting.

At first glance, this goal appears attainable via a simple mixing strategy.
In particular, one may choose, at each round $t$, a baseline usage probability $\epsilon_t$, and use the baseline with probability $\epsilon_t$, while otherwise deploying an empirically optimal online algorithm.
This suggests that one might be able to balance guarantee and regret by carefully tuning $(\epsilon_t)_{t \in [T]}$.
Such mixing strategies capture many standard algorithms, including explore-then-commit (ETC), which uses the baseline up to some round $m$ and commits thereafter, $\epsilon$-greedy strategies that use the baseline with a fixed probability at each round, and strategies that always deploy the empirically best-performing algorithm based on past observations.

However, \Cref{prop:naive-failure} shows that such a natural approach is fundamentally insufficient.

\begin{proposition}\label{prop:naive-failure}
There exist a repeated optimal stopping problem and a sequence of baseline algorithms $(g_t)_{t \in [T]}$ for $(\DD,p;T)$ such that the following holds:

Let $h_t$ be the empirically optimal online algorithm for $(\DD,p;T)$ based on $t-1$ samples.
For any sequence of baseline usage probabilities $(\epsilon_t)_{t \in [T]}$, define $h'_t$ as the algorithm that uses $g_t$ with probability $\epsilon_t$ and $h_t$ otherwise at round $t$.
Then, there exist distributions $\fD_+, \fD_- \in \DD$ such that at least one of the following \emph{does not} hold:
(i)~The regret of $(h'_t)_{t \in [T]}$ is $o(T)$ for $(\fD_+,p;T)$, and
(ii)~For at least $T/2$ rounds $t \in [T]$, $\Prob{h'_t(\fD_-) \ge g_t(\fD_-)} < 1/8$.
\end{proposition}

The proof is deferred to \Cref{app:pf-prop:naive-failure}.
This proposition implies that simple strategies based on pre-scheduling the baseline usage probabilities $(\epsilon_t)_{t \in [T]}$ and switching between the baseline and the empirically optimal online algorithm cannot achieve both sublinear regret and a high-probability guarantee with respect to the baseline, even under full-feedback settings.
Therefore, achieving both guarantees requires fundamentally more sophisticated strategies than such non-adaptive scheduling: one must simultaneously estimate the expected rewards of the baseline and the empirically optimal algorithm, and make decisions adaptively based on these estimates.

\subsection{Positive Results}
\label{subsec:our-results}

Having established in \Cref{sec:hardness} that balancing per-round guarantee and sublinear regret is theoretically viable only under full feedback with a high-probability relaxation, we now present our main positive result. Our main contribution is a general framework that achieves these objectives for the repeated optimal stopping problem.


\begin{theorem}[informal statement of \Cref{thm:main}]\label{thm:main-informal}
Let $g_t$ be an arbitrary online algorithm for $(\fD,p)$ with $t-1$ samples for each $t\in[T]$.
Suppose the profit function $p$ satisfies mild assumptions. 
Then we can construct $(h_t)_{t\in[T]}$ such that
(i)~$h_t=g_t$ for each $t\le B\sqrt{\kappa_{n,\Pi}T}$,
(ii)~$\Pr[h_t(\fD)\ge g_{\lfloor (t+1)/2\rfloor}(\fD)]\ge 1-O(t^{-\kappa_{n,\Pi}})$ for each $t\in[T]$,
and (iii)~$\regret=O(B\sqrt{\kappa_{n,\Pi}T\log T})$,
where $\kappa_{n,\Pi}$ depends only on the problem setting.
\end{theorem}

\Cref{thm:main-informal} has an important structural implication central to all our applications.
Our framework treats algorithm $g_t$ as a black box: $g_t$ can be \emph{any}
sample-based online algorithm for the single-round problem $(\fD,p)$, and our
construction preserves its per-round performance.
Consequently, if $g_t$ enjoys a strong per-round performance guarantee (e.g., in terms of competitive ratio), then the resulting sequence $(h_t)_{t\in[T]}$ preserves this baseline performance in each round while simultaneously achieving sublinear regret.
In addition, the theorem implies that the competitive ratio for the overall profit is asymptotically optimal, i.e.,
\begin{align}
     \frac{\Expec[(\bX^{(1)},\pi^{(1)}),\ldots, (\bX^{(T)},\pi^{(T)}) \simiid \fD]{\sum_{t=1}^T h_t[\History{t}](\fD;p)}}{T\cdot \optoff(\fD;p)} = (1-o(1))\cdot \frac{\opton(\fD;p)}{\optoff(\fD;p)}.
\end{align}

Our framework is applicable to a repeated setting on broad classes of optimal stopping problems.
We refer to a class of the repeated optimal stopping problem $(\DD,p;T)$ as ``the repeated X'' whenever $(\DD,p)$ corresponds to problem X.\footnote{In the literature, the term ``secretary'' is used in two different senses: (i) to indicate the \emph{random order} arrival model and (ii) to denote the \emph{best-choice} objective of selecting the overall maximum.}
We consider the following profit functions, each satisfying the mild assumptions required by our analysis: 
reward $\prwd(\bx,i)=x_i$, 
best choice $\pbest(\bx,i)=\Ind{x_i=\max_{j\in[n]}x_j}$, 
last success $\pls(\bx,i) = \Ind{i = \max\{i'\mid x_{i'}=1\}}$, and ski-rental $\pski(\bx,i) = \sum_{j=1}^{i-1}x_j+b$.
We set $p(\bx,n+1) = 0$ for $p \in \{\prwd, \pbest, \pls\}$.

Similarly, we define four distribution models: 
the adversarial order model $\Dadv$, 
the random order model $\Dro$, 
the i.i.d.\ distribution model $\Diid$, and 
the general order model $\Dall$ (where $\Pi$ is an arbitrary distribution over permutations of $[n]$).
%

For the repeated prophet inequality problem, represented as $(\Dadv,\prwd;T)$, our framework equipped with the single-sample algorithm~\cite{rubinstein_et_al:LIPIcs:2020:11745} as a baseline guarantees a competitive ratio of $1/n$ in round $1$ and $1/2$ with high probability (w.h.p.) for rounds $t\geq 2$. 
A similar guarantee holds for the repeated general order prophet inequality problem $(\Dall, \prwd; T)$.
Table~\ref{tab:applications} summarizes the results obtained from Theorem~\ref{thm:main} with existing baseline algorithms for well-known problems.
Full details are given in Appendix~\ref{sec:applications}.
We remark that, for any setting, the regret of our algorithm with respect to $T$ is bounded by $O(B\sqrt{\kappa_{n,\Pi} T\log T})$, where $\kappa_{n,\Pi}$ is a constant determined by the problem setting and independent of the number of rounds $T$.
For instance, when the distribution is $\Dadv$, we have $\kappa_{n,\Pi} = n$ and the regret is bounded by $O\bigl(B\sqrt{nT\log T}\bigr)$.

 
\begin{table}[t]
\centering
\small
\caption{Competitive ratios guaranteed by our framework.
All ratios for later rounds hold w.h.p.\
(with probability ${\geq}\, 1-O(t^{-\kappa_{n,\Pi}})$).
The regret in every case is $O(B\sqrt{\kappa_{n,\Pi}\, T\log T})$.
}
\label{tab:applications}
\begin{tabular}{@{}lccccc@{}}
\toprule
Problem $(\DD, p)$
  & $t=1$ & $t=2$ & $t\geq 3$ (w.h.p.)
  & $\kappa_{n,\Pi}$ & Baseline \\
\midrule
Prophet ineq.\ $(\Dadv,\prwd)$
  & $1/n$ & $1/2$ & $1/2$
  & $n$
  & \cite{rubinstein_et_al:LIPIcs:2020:11745} \\
Prophet sec.\ $(\Dro,\prwd)$
  & $1/e$ & $1/2$ & $0.688{-}O(t^{-1/5})$
  & $n!$
  & \cite{ferguson1989solved,rubinstein_et_al:LIPIcs:2020:11745,cristi2024prophet} \\
I.i.d.\ prophet ineq.\ $(\Diid,\prwd)$
  & $1/e$ & $1{-}1/e$ & $0.745{-}O(t^{-1})$
  & $n$
  & \cite{correa2019prophet,correa2024sample} \\
Adv.\ order sec.\ $(\Dadv,\pbest)$
  & $1/n$ & $1/4$ & $1/4$
  & $n$
  & \cite{DBLP:conf/soda/NutiV23} \\
R.o.\ secretary $(\Dro,\pbest)$
  & $1/e$ & $0.5009$ & $0.5009$
  & $n!$
  & \cite{DBLP:conf/soda/NutiV23} \\
Last success $(\Dall,\pls)$
  & $1/n$ & $1/4$ & $1/4$
  & $n$
  & \cite{yoshinaga2024last} \\
Ski-rental $(\Dall,\pski)$
  & $e/(e{-}1)$ & $e/(e{-}1)$ & $e/(e{-}1)$
  & $n$
  & \cite{karlin1994competitive} \\
\bottomrule
\end{tabular}
\end{table}

We also establish a lower bound of $\Omega(\sqrt{T})$ on the regret when $n=2$, the profit function is $\prwd$ or $\pbest$, and the distribution includes $\Diid$. 
Therefore, by ignoring $\kappa_{n,\Pi}$, our regret bound is tight with respect to $T$, up to a logarithmic factor.

Furthermore, for the repeated prophet inequality problem $(\Dadv,\prwd;T)$, we also obtain the following result, which improves the dependence of the regret on $\kappa_{n,\Pi}$ at the cost of the probability of guaranteeing the baseline performance.

\begin{theorem}[informal statement of~\Cref{thm:PI-refined-restate}]
\label{thm:PI-refined}
    For the repeated prophet inequality problem $(\Dadv,\prwd;T)$ and any $\fD \in \Dadv$, let $g_t$ be any sample-based algorithm. 
    Then, we can compute a sequence of algorithms $(h_t)_{t\in[T]}$ that achieves the following:
    (i)~for each $t \le \lfloor\sqrt{T}\rfloor$, $h_t = g_t$,
    (ii)~for each $t > \lfloor\sqrt{T}\rfloor$, with probability at least $1-2/t$, $h_t(\fD) \ge g_{\lfloor (t+1)/2\rfloor}(\fD)$, and
    (iii)~the regret of $(h_t)_{t\in[T]}$ is $O(\sqrt{T}\log T)$.
\end{theorem}

By combining the lower bound $\Omega(\sqrt{T})$ with \Cref{thm:PI-refined}, our result for the repeated prophet inequality guarantees the best possible performances.

\section{Preliminaries}\label{sec:preliminaries}



In this paper, we assume that the profit function $p$ satisfies the following two properties:
for any $i,j\in[n]$ and $x_1,\dots,x_n,x_i'\in[0,1]$ with $x_i\le x'_i$ and $j>i$, 
\begin{enumerate}[label=\textbf{(P\arabic*)}, ref=P\arabic*, leftmargin=*]
\item \label{prop:p-i} $p((x_1,\dots,x_i,\dots,x_n),i)\le p((x_1,\dots,x'_i,\dots,x_n),i)$
\item \label{prop:p-ii} $p((x_1,\dots,x_i,\dots,x_n),j)\ge p((x_1,\dots,x'_i,\dots,x_n),j)$ 
\end{enumerate}
These conditions are satisfied by the profit functions of standard optimal stopping problems: expected reward $\prwd$, best choice $\pbest$, last success $\pls$, and ski-rental $\pski$.
Note that the maximum profit value $B$ can be set to $1$ for the expected reward, best choice, last success, and can be set to $n+b$ for the ski-rental setting.
An example of a profit function that does not satisfy properties (\ref{prop:p-i}) and (\ref{prop:p-ii}) is the scenario in which one aims to maximize the probability of selecting the second-largest value~\cite{vanderbei2021postdoc}.

Next, we define several classes of distributions \( \fD=(D_1,\dots,D_n;\Pi) \).
Let \( \Dall \) be the set where \( \Pi \) is any distribution over permutations of \( [n] \).
Let \( \Dadv \) be the set where \( \Pi \) is the point mass at the identity permutation, denoted by \( \Piid \).
Let \( \Dro \) be the set where \( \Pi \) is the uniform distribution over all permutations.
Let \( \Diid \) be the set where \( \Pi=\Piid \) and \( D_1=\cdots=D_n \).
Finally, let the forward-backward class be the set where \( \Pi \) is uniform over the identity permutation and its reverse.
It holds that \( \Diid\subsetneq\Dadv\subsetneq\Dall \) and \( \Diid\subsetneq\Dro\subsetneq\Dall \).






\subsection{Threshold Algorithm}\label{sec:threshold}

In this subsection, we show that any optimal online algorithm is a threshold algorithm.
A threshold algorithm is an algorithm $h = (f_i)_{i=1}^n$ specified by a sequence of functions $f_i\colon [0,1]^{i}\times\Inj([i],[n]) \to \{0,1\}$ for each $i\in[n]$, such that
$f_i(\bx_{\le i},\tau_{\le i}) = 1$ if $x_i>\theta_i(\tau_{\le i})$ and $0$ if $x_i<\theta_i(\tau_{\le i})$,
where $\theta_i \colon \Inj([i],[n]) \to [0,1]\cup\{\infty\}$ is a threshold function that depends only on the order of first $i$ values $\tau_{\le i}$. 
If $x_i=\theta_i(\tau_{\le i})$, the value of $f_i(\bx_{\le i},\tau_{\le i})$ can be either $0$ or $1$.
$\theta_i$ may take value $\infty$, meaning the $i$th variable is never accepted.

\begin{lemma}\label{lem:threshold}
Suppose that the profit function $p$ satisfies properties (\ref{prop:p-i}) and (\ref{prop:p-ii}).
Furthermore, assume that the following condition holds:
\begin{enumerate}[label=\textbf{(A\arabic*)}, ref=A\arabic*, font=\upshape, leftmargin=*]
    \item \label{assumption:q-exists}
    for each $i\in[n]$, there exists a function $q_i\colon [0,1]^{n-i+1}\times\{i,\dots,n+1\}\to\mathbb{R}$ such that $q_i(\bx_{\ge i},j)=p(\bx,j)-p(\bx,i)$ for any $j\in\{i,\dots,n+1\}$ and $\bx\in[0,1]^n$.
\end{enumerate}
Then, there exists a threshold algorithm $h^* =(f^*_i)_{i=1}^n$ such that $\opton(\fD,p) = h^*(\fD,p)$.
\end{lemma}

The proof is deferred to \Cref{subsec:proof threshold}.
In the proof, we show a setting of a threshold $\theta_i(\tau_{\le i})$ using the assumption (\ref{assumption:q-exists}).

We remark that the threshold $\theta_i(\tau_{\le i})$ generally depends not only on the set $\{\tau(1),\dots,\tau(i-1)\}$ but also on the specific order $\tau_{\le i-1}$ of the first $i-1$ variables. This is because the conditional distribution of the remaining permutations may depend on the observed order so far. However, when $\Pi$ is the uniform distribution over all permutations (i.e., the random order model), the conditional distribution is independent of the observed order, and thus it suffices for the threshold to depend only on $\{\tau(1),\dots,\tau(i-1)\}$ and $\tau(i)$.

\begin{corollary}
If the profit function $p$ is given by expected reward, last success, or ski-rental, 
then there exists a threshold algorithm that is an optimal online algorithm. 
\end{corollary}

Unfortunately, the profit function $p$ for best choice $p(\bx,i) = \Ind{x_i = \max_{j \in [n]} x_j}$ does not satisfy the assumption (\ref{assumption:q-exists}).
However, we can construct an optimal online algorithm that is almost a threshold algorithm.
If $x_i\ge \max_{j\in[i-1]}x_j$, then $p(\bx,j)-p(\bx,i)$ depends on $\bx_{\geq i}$ and $j$, and hence we can define $q_i(\bx_{\geq i}, j) = p(\bx,j)-p(\bx,i)$ for any $j \geq i$.
We can set a threshold $\theta_i(\tau_{\leq i})$ for such a case (see \Cref{subsec:proof threshold} for details).
If $x_i<\max_{j\in[i-1]}x_j$, then there is no loss in rejecting the $i$th variable.
Thus, the following claim holds.
\begin{corollary}
If the profit function $p$ is given by best choice, then there exists an optimal online algorithm $h=(f_i)_{i=1}^n$ and a threshold algorithm $h'=(f_i')_{i=1}^n$ such that
\begin{align}
f_i(\bx_{\le i},\tau_{\le i})=\begin{cases}
f'_i(\bx_{\le i},\tau_{\le i}) & \text{if }x_i\ge \max_{j\in[i-1]}x_j,\\
0                              & \text{if }x_i< \max_{j\in[i-1]}x_j.
\end{cases}
\end{align}
\end{corollary}
We call the function defined in this corollary an \emph{essentially threshold} algorithm.
In the following sections, we focus on (essentially) threshold algorithms in analysis of the regret.

\subsection{Uniform Concentration Bound}\label{sec:uniform laws}

We review the uniform law of large numbers used to derive regret bounds.
We recall the Rademacher complexity.


\begin{definition}\label{def:rademacher}
    Let $A$ be a subset of $\R^t$. The \emph{Rademacher complexity} of $A$ is defined as
    \[
        \rade(A) \coloneqq \E_{\sigma} \left[ \frac{1}{t} \cdot \sup_{(a_1,\ldots,a_t) \in A} \left| \sum_{s=1}^t \sigma_s a_s \right| \right],
    \]
    where $\sigma_1, \ldots, \sigma_t$ are independent Rademacher random variables.
    Let $\cHth$ be a class of functions.
    The \emph{empirical Rademacher complexity} of $\cHth$ on samples $(\bx^{(1)},\ldots, \bx^{(t)})$ is
    \[
        \rade_{(\bx^{(1)},\ldots, \bx^{(t)})}(\cHth) \coloneqq  \E_{\sigma} \left[ \frac{1}{t} \sup_{h \in \cHth} \left| \sum_{s=1}^t \sigma_s h(\bx^{(s)}) \right| \right].
    \]
    Furthermore, the \emph{Rademacher complexity} of $\cHth$ for $t$ i.i.d.\ samples from $\fD$ is defined as
    \[
        \rade_{t,\fD}(\cHth) \coloneqq \E_{\bX^{(1)},\ldots, \bX^{(t)} \simiid \fD} \left[ \rade_{(\bX^{(1)},\ldots, \bX^{(t)})}(\cHth) \right].
    \]
\end{definition}

The following lemma is an immediate consequence of Massart's Lemma~\cite[Lemma 5.2]{massart2000some} and the proof is given in \Cref{subsec:massart}.


\begin{lemma}
\label{lem:massart}
For any finite subset $A\subset\R^t$, $\rade(A) \le \max_{a\in A}\|a\|_2 \sqrt{2 \ln (2|A|)}/t$.
\end{lemma}


\begin{theorem}[Uniform Concentration Bound{~\cite[Theorem 4.10]{Wainwright_2019}}]
\label{thm:ulln}
    Let $\cHth$ be a function class such that $\|h\|_{\infty} \le B$ for all $h \in \cHth$. For any $\eta \ge 0$ and positive integer $t$, we have
    \begin{align}
        \Prob[\bX^{(1)},\dots,\bX^{(t)}\simiid\fD]{\sup_{h \in \cHth} \left| \E_{\bX \sim \fD}[h(\bX)] - \frac{1}{t} \sum_{s=1}^t h(\bX^{(s)}) \right| > 2\rade_{t,\fD}(\cHth) + \eta}
        \le 2 e ^{-\frac{t\eta^2}{2B^2}}.
    \end{align}
\end{theorem}


\section{General Result}\label{sec:general}

We present a general framework for repeated optimal stopping $(\DD,p;T)$ that achieves both per-round guarantees and sublinear regret.

Starting from competitive algorithms $g_t$ based on $t-1$ samples, we construct algorithms $h_t$ that ensure low regret via uniform concentration and problem-specific techniques.
We then select between $g_t$ and $h_t$ using a hold-out estimate, since the underlying distribution is unknown and their expected rewards cannot be evaluated directly.
We split the first $t-1$ samples into training and test sets.
Using $\zeta(t)-1$ samples for training, we compare $g_{\zeta(t)}$ and $h_{\zeta(t)}$ via their empirical performance on the remaining samples.

\begin{remark}
By the uniform concentration bound (\Cref{thm:ulln}), we can evaluate the performance of $g_t$ and $h_t$ even when all samples are used as training data. 
This allows us to achieve a reward no worse than that of $g_t$ w.h.p.\ in each round $t$.
However, to obtain a meaningful regret bound under this approach, $g_t$ needs to belong to a small algorithm class, such as (essentially) threshold algorithms.
In contrast, our method based on the hold-out approach does not require such restrictions on $g_t$.
\end{remark}


\subsection{Switching rule}



For each $t\in[T]$, assume we have an algorithm $h_t$ such that there exist positive reals $\epsilon_1(t)$ and $\delta_1(t)$ satisfying the following: $h_t$ is an $\epsilon_1(t)$-additive approximation of $\opton$ with probability at least $1-\delta_1(t)$; that is,
\begin{align}
\label{eq:sample-complexity}
    \Prob[(\bX^{(1)},\pi^{(1)}),\dots,(\bX^{(t-1)},\pi^{(t-1)}) \simiid \fD]{h_t[\History{t}](\fD) \ge \opton(\fD) - \epsilon_1(t)} \ge 1-\delta_1(t).
\end{align}
We will construct such algorithms in the next subsection.
The parameters $\epsilon_1(t)$ and $\delta_1(t)$ are directly related to the \emph{sample complexity} of $(\fD,p)$. The sample complexity of $(\fD,p)$ is defined as the minimum number of i.i.d.\ samples required to find an $\epsilon$-additive approximation of $\opton$ with probability at least $1-\delta$.
The parameters $\epsilon_1(t)$ and $\delta_1(t)$ correspond to the accuracy and confidence levels achieved when the sample complexity of $(\fD, p)$ is $t$. 

When both $\epsilon_1(\zeta(t))$ and $\delta_1(\zeta(t))$ are sufficiently small, applying algorithm $h_{\zeta(t)}$ in each round achieves sublinear regret. 
In each round $t\in[T]$, we choose between executing $h_{\zeta(t)}$ and $g_{\zeta(t)}$ based on their estimated performance.

In each round $t$ with $\zeta(t)<t$, we estimate the expected reward $h_{\zeta(t)}[\History{\zeta(t)}](\fD)$ of algorithm $h_{\zeta(t)}$, which depends on the first $\zeta(t) - 1$ samples, as
\begin{align}
\label{eq:estimate}
    \hat{h}_{\zeta(t)}(\History{t};\zeta(t))
    \coloneqq \frac{1}{t-\zeta(t)} \sum_{s=\zeta(t)}^{t-1} h_{\zeta(t)}[\History{\zeta(t)}](\bX^{(s)},\pi^{(s)};p).
\end{align}
This estimate is simply the arithmetic mean of the rewards obtained by $h_{\zeta(t)}$ on the test samples.
We emphasize that this partitioning ensures the independence between the algorithm parameters and the test samples, so that the random variables $h_{\zeta(t)}[\History{\zeta(t)}](\bX^{(s)}, \pi^{(s)}; p)$ for $s=\zeta(t),\zeta(t)+1,\dots,t-1$ are mutually independent. Consequently, Hoeffding's inequality can be applied. 
Since~\eqref{eq:estimate} is the mean of $(t-\zeta(t))$ independent random variables, each taking values in $[0, B]$, for any $\eta > 0$ we have
\begin{align}
\label{eq:approx-reward}
    \Prob{\Big|h_{\zeta(t)}[\History{\zeta(t)}](\fD) - \hat{h}_{\zeta(t)}(\History{t};\zeta(t))\Big| > \eta} \le \delta_0(t-\zeta(t),\eta),
\end{align}
where
$\delta_0(t-\zeta(t),\eta) \coloneqq 2\cdot\exp\left( -{2(t-\zeta(t))\eta^2}/{B^2} \right)$.
For each $t$, define $\cE_t$ as the event in which all of the following conditions are met:
\begin{align}
    \label{eq:def-event-E-t}
    \cE_t: \qquad
    \begin{aligned}
        &|g_{\zeta(t)}[\History{\zeta(t)}](\fD) - \hat{g}_{\zeta(t)}(\History{t};\zeta(t))|&\le \epsilon(t), \\
        &|h_{\zeta(t)}[\History{\zeta(t)}](\fD) - \hat{h}_{\zeta(t)}(\History{t};\zeta(t))|&\le \epsilon(t), \\
        &|\opton(\fD) - h_{\zeta(t)}[\History{\zeta(t)}](\fD)| &\le \epsilon(t),
    \end{aligned}
\end{align}
where $\epsilon(t)\coloneq \epsilon_1(\zeta(t))$.
The event $\cE_t$ indicates the success in estimating the expected reward and approximating the optimal online algorithm.
By~\eqref{eq:approx-reward}, \eqref{eq:sample-complexity} and the union bound, $\cE_t$ occurs with probability at least $1-\delta(t)$, where $\delta(t)\coloneqq 2\delta_0\big(t-\zeta(t),\epsilon_1(\zeta(t))\big) + \delta_1\big(\zeta(t)\big)$.
Next, define another event $\cC_t$ as 
\begin{align}\label{eq:event-C}
    \cC_t: \quad
\hat{g}_{\zeta(t)}(\History{t};\zeta(t)) + \epsilon(t) \le \hat{h}_{\zeta(t)}(\History{t};\zeta(t))- \epsilon(t).
\end{align}
The event $\cC_t$ indicates that the estimated reward of $h_{\zeta(t)}$ is sufficiently greater than that of $g_{\zeta(t)}$.
If $t=\zeta(t)$, the estimated values are not well-defined, but we regard the event $\cC_t$ as not having occurred.

For each round $t$, we define $h_t^*$ as the algorithm that uses $h_{\zeta(t)}[\History{\zeta(t)}]$ when the event $\cC_t$ occurs and $g_{\zeta(t)}[\History{\zeta(t)}]$ otherwise.
The next theorem shows that $h_t^*$ achieves at least the expected performance of $g_t$ and the expected regret is bounded (see \Cref{subsec:pf-thm:uniform-comp-and-regret} for the proof).

\begin{theorem}
\label{thm:uniform-comp-and-regret}
    For a round $t\in[T]$ with $\zeta(t)<t$, suppose that $h_{\zeta(t)}$ satisfies \eqref{eq:sample-complexity}.
    Then, we have
    \begin{enumerate}
        \item $h_t^*(\fD) \ge g_{\zeta(t)}(\fD)$
        when $\cE_t$ occurs (i.e., with probability at least $1-\delta(t)$),
        \item $\Expec{\opton(\fD) - h_t^*(\fD)} \le 6\epsilon(t) + B\delta(t)$ (the expected regret per round is bounded).
    \end{enumerate}
\end{theorem}

\Cref{thm:uniform-comp-and-regret} does not provide a method for finding $h_t$ that satisfies equation~\eqref{eq:sample-complexity}. Therefore, next we will provide a general method for finding $h_t$ that satisfies~\eqref{eq:sample-complexity} and a specific functional form for $\epsilon_1(t)$ and $\delta_1(t)$.

\subsection{Regret bound}

Let $\cHth$ be a set of algorithms such that, for any $\fD\in\DD$, $\opton(\fD)=h(\fD;p)$ for some $h\in\cHth$.
For example, if $p$ satisfies the assumption (\ref{assumption:q-exists}) in \Cref{lem:threshold}, then $\cHth$ can be chosen as the set of threshold algorithms.
Also, if $p$ is a profit function of best choice, then $\cHth$ can be chosen as the set of essentially threshold algorithms.
Let $\hat{\fD}^{(t)}$ be the empirical distribution where $(\bX^{(s)},\pi^{(s)})$ is sampled with probability $1/(t-1)$ for each $s\in[t-1]$.
Thus, for any algorithm $h\in\cHth$, the expected profit of $h$ for the empirical distribution $\hat{\fD}^{(t)}$ at round $t$ is defined as
\begin{align}
h(\hat{\fD}^{(t)};p) = \frac{1}{t-1}\sum_{s=1}^{t-1} h(\bX^{(s)},\pi^{(s)};p).
\end{align}
For a distribution $\fD'$ that is either $\fD$ or $\hat{\fD}^{(t)}$, we denote by $h_{\fD'}$ an optimal algorithm in $\cHth$ for $\fD'$, i.e., $h_{\fD'}\in\argmax_{h\in\cHth} h(\fD'; p)$. It should be noted that $h_{\fD'}$ is not necessarily an optimal online algorithm for $\fD'$ because an empirical distribution may not be in $\DD$.
Moreover, the set $\argmax_{h\in\cHth} h(\fD'; p)$ is nonempty. This is because, if $\fD'=\fD$, then $\opton(\fD)=h(\fD;p)$ for some $h\in\cHth$ by the assumption. If $\fD'=\hat{\fD}^{(t)}$, then we can partition $\cHth$ into at most $(n+1)^{t-1}$ classes according to the acceptance index of each realization in the first $t-1$ rounds. Because this yields only finitely many possibilities, an optimal algorithm in $\cHth$ must exist.
To derive explicit forms of $\epsilon_1(t)$ and $\delta_1(t)$, we combine a uniform concentration bound with an upper bound on the number of distinct behaviors of algorithms in $\cHth$ on the observed samples.

\begin{lemma}\label{col:delta-epsilon}
    Suppose that $\cHth$ is the set of threshold algorithms or the set of essentially threshold algorithms.
    Then, $\epsilon_1(t)=6B\sqrt{(2\min(n|\Pi|,en!) \ln (4t))/(t-1)}$ and $\delta_1(t)=1/(2t)^{\min(n|\Pi|,en!)}$ satisfy the assumption \eqref{eq:sample-complexity} for any $t\in[T]$ with $t\ge 2$.
\end{lemma}

By combining \Cref{thm:uniform-comp-and-regret} and \Cref{col:delta-epsilon}, we obtain our main theorem which provides a general method of constructing a selection rule that guarantees both a competitive ratio in each round and a sublinear regret. 

\begin{theorem}\label{thm:main}
    Let $(\fD,p;T)$ be an instance of the repeated optimal stopping problem $(\DD,p;T)$, and let $\cHth$ be the set of (essentially) threshold algorithms.
    Suppose that $\opton(\fD;p)=h(\fD;p)$ for some $h\in\cHth$.
    Let $\zeta(t) \coloneqq \lfloor (t+1)/2 \rfloor$.
    
    For each $t\in[T]$, let $g_t[\History{t}]$ be an algorithm for the optimal stopping problem $(\fD,p)$ with $t-1$ samples and $h_t[\History{t}]$ be the optimal algorithm among $\cHth$ with respect to the empirical distribution $\hat{\fD}^{(t)}$.
    For a parameter $t_0\ge 1$, let $h_t^*=g_t$ for $t \le t_0$ and, for $t > t_0$,
    let $h_t^*$ be the algorithm that uses $h_{\zeta(t)}$ when the event $\cC_t$ occurs and $g_{\zeta(t)}$ otherwise,
    where we set
    \begin{align}
        \epsilon_1(t) \coloneqq 6B \sqrt{\frac{2 \min(n|\Pi|,en!) \ln (4t)}{(t-1)}} \quad\text{and}\quad
        \delta_1(t) \coloneqq \frac{1}{(2t)^{\min(n|\Pi|,en!)}}.
    \end{align}
    Then, the sequence of online algorithms $(h_t^*)_{t\in [T]}$ satisfies the following:
    \begin{itemize}
        \item For each $t > t_0$, with probability at least $1-2/t^{\min(n|\Pi|,en!)}$, we have
        $$h_t^*(\fD)\ge g_{\zeta(t)}[\History{\zeta(t)}](\fD).$$
        \item The regret of the algorithm $(h_t^*)_{t\in[T]}$ is $O(Bt_0 + B\sqrt{\min(n|\Pi|,n!)T\log T})$.
    \end{itemize}
\end{theorem}

The proof of \Cref{thm:main} is given in~\Cref{subsec:pf-main}.
Since $|\Pi|=1$ for $\Dadv$ and $\Diid$,
the regret bound of
$(h_t^*)_{t\in[T]}$ constructed in \Cref{thm:main} is
$O(B\sqrt{nT\log T})$ for $\Dadv$ and $\Diid$, and
$O(B\sqrt{n! T\log T})$ for $\Dro$ and $\Dall$.\footnote{
For the random order model, a refined analysis yields the improved regret bound $O(B \sqrt{n 2^{n-1} T \log T})$; see \Cref{subsubsec:refined-random-order} for details.
}


\section{Lower Bound}

In this section, we establish a regret lower bound of $\Omega(\sqrt{T})$ for optimal stopping problems. The proof is given in Appendix~\ref{subsec:pf-lower-bound}.

\begin{theorem}
\label{thm:lower-bound}
    When $n=2$, any algorithm for both $(\Diid,\prwd;T)$ and $(\Diid,\pbest;T)$ incurs a regret of $\Omega(\sqrt{T})$.
\end{theorem}

We have shown an upper bound of $O(\sqrt{T\log T})$ in \Cref{thm:main}.
Therefore, ignoring the dependency of $n$ and $|\Pi|$, this lower bound is tight up to a logarithmic factor with respect to the number of rounds $T$.
Specifically, by \Cref{thm:PI-refined}, the lower bound $\Omega(\sqrt{T})$ is nearly tight for the repeated (adversarial / i.i.d.) prophet inequality problem.

Since $\Diid$ is a special case of $\Dadv$, $\Dro$, and $\Dall$, the same lower bound holds for the corresponding problems with these distribution families.

Our results strengthen and generalize the recent finding of Gatmiry et al.~\cite[Theorem 5.1]{Gatmiry2024}, who established a lower bound $\Omega(\sqrt{T})$ for the repeated prophet inequality $(\Dadv,\prwd;T)$ with $n=2$ \emph{distinct} distributions.
In contrast, we establish the same lower bound for the case of $n=2$ \emph{identical} distributions. Consequently, our results apply not only to the adversarial order model but also to the random order and i.i.d.\ models.
Moreover, since Gatmiry et al.~\cite{Gatmiry2024} did not consider the best-choice profit $\pbest$, our result is the first to establish this lower bound in the context of optimal stopping problems involving $\pbest$.



\bibliography{ref}

@inproceedings{andrew2013tale,
  title        = {A tale of two metrics: Simultaneous bounds on competitiveness and regret},
  author       = {Andrew, Lachlan and Barman, Siddharth and Ligett, Katrina and Lin, Minghong and Meyerson, Adam and Roytman, Alan and Wierman, Adam},
  booktitle    = {Proceedings of the 26th Annual Conference on Learning Theory},
  pages        = {741--763},
  year         = {2013}
}

@inproceedings{azar2014prophet,
  title     = {Prophet inequalities with limited information},
  author    = {Azar, Pablo D. and Kleinberg, Robert and Weinberg, S. Matthew},
  booktitle = {Proceedings of the 25th Annual ACM-SIAM Symposium on Discrete Algorithms},
  pages     = {1358--1377},
  year      = {2014}
}

@article{Buchbinder2016,
  author  = {Buchbinder, Niv and Chen, Shahar and Naor, Joseph (Seffi) and Shamir, Ohad},
  title   = {Unified Algorithms for Online Learning and Competitive Analysis},
  journal = {Mathematics of Operations Research},
  volume  = {41},
  number  = {2},
  pages   = {612--625},
  year    = {2016},
}

@inproceedings{caramanis2022single,
  title        = {Single-sample prophet inequalities via greedy-ordered selection},
  author       = {Caramanis, Constantine and D{\"u}tting, Paul and Faw, Matthew and Fusco, Federico and Lazos, Philip and Leonardi, Stefano and Papadigenopoulos, Orestis and Pountourakis, Emmanouil and Reiffenh{\"a}user, Rebecca},
  booktitle    = {Proceedings of the 2022 Annual ACM-SIAM Symposium on Discrete Algorithms},
  pages        = {1298--1325},
  year         = {2022}
}

@inproceedings{correa2019prophet,
  title     = {Prophet inequalities for i.i.d. random variables from an unknown distribution},
  author    = {Correa, Jos{\'e} R. and D{\"u}tting, Paul and Fischer, Felix and Schewior, Kevin},
  booktitle = {Proceedings of the 20th ACM Conference on Economics and Computation},
  pages     = {3--17},
  year      = {2019}
}

@inproceedings{correa2020two,
  title     = {The two-sided game of googol and sample-based prophet inequalities},
  author    = {Correa, Jos{\'e} R. and Cristi, Andr{\'e}s and Epstein, Boris and Soto, Jos{\'e} A.},
  booktitle = {Proceedings of the 31st ACM-SIAM Symposium on Discrete Algorithms},
  pages     = {2066--2081},
  year      = {2020}
}

@inproceedings{correa2021unknown,
  author    = {Correa, Jos\'{e} and D\"{u}tting, Paul and Fischer, Felix and Schewior, Kevin and Ziliotto, Bruno},
  title     = {Unknown I.I.D.\ Prophets: Better Bounds, Streaming Algorithms, and a New Impossibility},
  booktitle = {Proceedings of the 12th Innovations in Theoretical Computer Science Conference},
  pages     = {86:1},
  year      = {2021},
}

@inproceedings{DBLP:conf/soda/NutiV23,
  author    = {Pranav Nuti and Jan Vondr{\'{a}}k},
  title     = {Secretary Problems: The Power of a Single Sample},
  booktitle = {Proceedings of the 2023 {ACM-SIAM} Symposium on Discrete Algorithms},
  pages     = {2015--2029},
  year      = {2023},
}

@article{ferguson1989solved,
  title     = {Who solved the secretary problem?},
  author    = {Ferguson, Thomas S},
  journal   = {Statistical Science},
  volume    = {4},
  number    = {3},
  pages     = {282--289},
  year      = {1989},
  publisher = {Institute of Mathematical Statistics}
}

@inproceedings{Gatmiry2024,
  author    = {Khashayar Gatmiry and Thomas Kesselheim and Sahil Singla and Yifan Wang},
  title     = {Bandit Algorithms for Prophet Inequality and {Pandora's} Box},
  booktitle = {Proceedings of the 2024 Annual ACM-SIAM Symposium on Discrete Algorithms},
  chapter   = {},
  pages     = {462--500},
  year      = {2024}
}

@inproceedings{Goel2023,
  title     = {Best of Both Worlds in Online Control: Competitive Ratio and Policy Regret},
  author    = {Goel, Gautam and Agarwal, Naman and Singh, Karan and Hazan, Elad},
  booktitle = {Proceedings of the 5th Annual Learning for Dynamics and Control Conference},
  pages     = {1345--1356},
  year      = {2023},
}

@inproceedings{kaplan2020competitive,
  author    = {Kaplan, Haim and Naori, David and Raz, Danny},
  title     = {Competitive Analysis with a Sample and the Secretary Problem},
  booktitle = {Proceedings of the 31st Annual ACM-SIAM Symposium on Discrete Algorithms},
  year      = {2020},
  pages     = {2082--2095}
}

@article{Krengel1977-uu,
  author  = {Ulrich Krengel and Louis Sucheston},
  title   = {Semiamarts and finite values},
  journal = {Bulletin of the American Mathematical Society},
  year    = {1977},
  volume  = {83},
  number  = {4},
  pages   = {745--747}
}

@article{massart2000some,
  title     = {Some applications of concentration inequalities to statistics},
  author    = {Massart, Pascal},
  journal   = {Annales de la Facult{\'e} des sciences de Toulouse: Math{\'e}matiques},
  volume    = {9},
  number    = {2},
  pages     = {245--303},
  year      = {2000}
}

@inproceedings{pmlr-v134-guo21a,
  title     = {Generalizing Complex Hypotheses on Product Distributions: Auctions, Prophet Inequalities, and {Pandora's} Problem},
  author    = {Guo, Chenghao and Huang, Zhiyi and Tang, Zhihao Gavin and Zhang, Xinzhi},
  booktitle = {Proceedings of the 34th Conference on Learning Theory},
  pages     = {2248--2288},
  year      = {2021},
}

@inproceedings{pmlr-v98-daniely19a,
  title     = {Competitive ratio vs regret minimization: Achieving the best of both worlds},
  author    = {Daniely, Amit and Mansour, Yishay},
  booktitle = {Proceedings of the 30th International Conference on Algorithmic Learning Theory},
  pages     = {333--368},
  year      = {2019},
}

@inproceedings{rubinstein_et_al:LIPIcs:2020:11745,
  author    = {Aviad Rubinstein and Jack Z. Wang and S. Matthew Weinberg},
  title     = {Optimal Single-Choice Prophet Inequalities from Samples},
  booktitle = {Proceedings of the 11th Innovations in Theoretical Computer Science Conference},
  pages     = {60:1--60:10},
  year      = {2020},
}

@inbook{Wainwright_2019,
  place      = {Cambridge},
  series     = {Cambridge Series in Statistical and Probabilistic Mathematics},
  title      = {Uniform laws of large numbers},
  booktitle  = {High-Dimensional Statistics: A Non-Asymptotic Viewpoint},
  publisher  = {Cambridge University Press},
  author     = {Wainwright, Martin J},
  year       = {2019},
  pages      = {98--120},
}

@inproceedings{yoshinaga2024last,
  author    = {Yoshinaga, Toru and Kawase, Yasushi},
  title     = {The Last Success Problem with Samples},
  booktitle = {Proceedings of the 32nd Annual European Symposium on Algorithms},
  pages     = {105:1--105:15},
  year      = {2024},
}

@article{dynkin1963optimum,
  title={The optimum choice of the instant for stopping a {Markov} process},
  author={Dynkin, Evgenii Borisovich},
  journal={Soviet Mathematics},
  volume={4},
  pages={627--629},
  year={1963}
}

@article{vanderbei2021postdoc,
  title={The postdoc variant of the secretary problem},
  author={Vanderbei, Robert J},
  journal={Mathematica Applicanda},
  volume={49},
  number={1},
  year={2021}
}

@article{karlin1994competitive,
  title={Competitive randomized algorithms for nonuniform problems},
  author={Karlin, Anna R and Manasse, Mark S and McGeoch, Lyle A and Owicki, Susan},
  journal={Algorithmica},
  volume={11},
  number={6},
  pages={542--571},
  year={1994},
  publisher={Springer}
}

@inproceedings{Shah2021,
author = {Shah, Anant and Rajkumar, Arun},
title = {Sequential Ski Rental Problem},
year = {2021},
booktitle = {Proceedings of the 20th International Conference on Autonomous Agents and MultiAgent Systems},
pages = {1173--1181},
numpages = {9},
}

@inproceedings{agarwal2024semi,
  title={Semi-bandit learning for monotone stochastic optimization},
  author={Agarwal, Arpit and Ghuge, Rohan and Nagarajan, Viswanath},
  booktitle={Proceedings of the 65th Annual Symposium on Foundations of Computer Science},
  pages={1260--1274},
  year={2024},
}

@inproceedings{jin2024sample,
  title={Sample complexity of posted pricing for a single item},
  author={Jin, Billy and Kesselheim, Thomas and Ma, Will and Singla, Sahil},
  booktitle={Proceedings of the 38th International Conference on Neural Information Processing Systems},
  pages={82296--82317},
  year={2024}
}

@article{liu2025improved,
  title={Improved Regret and Contextual Linear Extension for Pandora's Box and Prophet Inequality},
  author={Liu, Junyan and Chen, Ziyun and Wang, Kun and Luo, Haipeng and Ratliff, Lillian J},
  journal={arXiv preprint arXiv:2505.18828},
  year={2025}
}

@inproceedings{fu2020learning,
  title={Learning utilities and equilibria in non-truthful auctions},
  author={Fu, Hu and Lin, Tao},
  booktitle={Proceedings of the 34th International Conference on Neural Information Processing Systems},
  pages={14231--14242},
  year={2020}
}

@inproceedings{guo2019settling,
  title={Settling the sample complexity of single-parameter revenue maximization},
  author={Guo, Chenghao and Huang, Zhiyi and Zhang, Xinzhi},
  booktitle={Proceedings of the 51st Annual ACM Symposium on Theory of Computing},
  pages={662--673},
  year={2019}
}

@inproceedings{Chen2025ProphetSecretary,
author    = {Jianli Chen and
Paul D{\"u}tting and
Evangelia Gergatsouli and
Rojin Rezvan and
Alexandros Tsigonias-Dimitriadis},
title     = {Prophet Secretary and Matching: the Significance of the Largest Item},
booktitle = {Proceedings of the 2025 Annual ACM-SIAM Symposium on Discrete Algorithms},
pages     =  {1371--1401}, 
year      = {2025}
}

@inproceedings{cristi2024prophet,
  title={Prophet inequalities require only a constant number of samples},
  author={Cristi, Andr{\'e}s and Ziliotto, Bruno},
  booktitle={Proceedings of the 56th Annual ACM Symposium on Theory of Computing},
  pages={491--502},
  year={2024}
}

@article{correa2024sample,
  title     = {Sample-driven optimal stopping: From the secretary problem to the i.i.d.\ prophet inequality},
  author    = {Correa, Jos{\'e} R. and Cristi, Andr{\'e}s and Epstein, Boris and Soto, Jos{\'e} A.},
  journal   = {Mathematics of Operations Research},
  volume    = {49},
  number    = {1},
  pages     = {441--475},
  year      = {2024}
}

@article{hill1982comparisons,
  title={Comparisons of stop rule and supremum expectations of iid random variables},
  author={Hill, Theodore P and Kertz, Robert P},
  journal={The Annals of Probability},
  pages={336--345},
  year={1982},
  publisher={JSTOR}
}

@article{kertz1986stop,
  title={Stop rule and supremum expectations of iid random variables: a complete comparison by conjugate duality},
  author={Kertz, Robert P},
  journal={Journal of Multivariate Analysis},
  volume={19},
  number={1},
  pages={88--112},
  year={1986},
  publisher={Elsevier}
}

@inproceedings{correa2017posted,
  title={Posted price mechanisms for a random stream of customers},
  author={Correa, Jos{\'e} and Foncea, Patricio and Hoeksma, Ruben and Oosterwijk, Tim and Vredeveld, Tjark},
  booktitle={Proceedings of the 2017 ACM Conference on Economics and Computation},
  pages={169--186},
  year={2017}
}

@inproceedings{liu2020variable,
  author       = {Allen Liu and
                  Renato Paes Leme and
                  Martin P{\'{a}}l and
                  Jon Schneider and
                  Balasubramanian Sivan},
  editor       = {P{\'{e}}ter Bir{\'{o}} and
                  Shuchi Chawla and
                  Federico Echenique},
  title        = {Variable Decomposition for Prophet Inequalities and Optimal Ordering},
  booktitle    = {Proceedings of the 2021 ACM Conference on Economics and Computation,},
  pages        = {692},
  year         = {2021}
}

\appendix
\markboth{Appendix}{Appendix}

\section{Applications}\label{sec:applications}

In this section, we demonstrate how \Cref{thm:main} can be applied to specific repeated optimal stopping problems $(\DD,p;T)$. 
We treat repeated versions of the prophet inequality, the prophet secretary problem, the i.i.d.\ prophet inequality, the adversarial order secretary problem, the random order secretary problem, and the i.i.d.\ secretary problem.

Throughout this section, each corollary has the following generic structure:
\begin{corollary}
\label{cor:generic}
Let $(g_t)_{t\in[T]}$ be the sequence of algorithms with provable competitive ratios against $\optoff(\fD)$ and let $(h_t^*)_{t\in[T]}$ be algorithms constructed in \Cref{thm:main} with baseline algorithms $(g_t)_{t\in[T]}$ and a suitable choice of $t_0$.
Then, there exists $\kappa > 0$ such that the sequence $(h_t^*)_{t\in[T]}$ of algorithms achieves
(i) $h_t^* = g_t$ for $t \le t_0$,
(ii) $h_t^*(\fD) \ge g_{\zeta(t)}(\fD)$ with probability at least $1-2/t^{\kappa}$ for $t>t_0$, and
(iii) sublinear regret $O(B\sqrt{\kappa T \log T})$. 
\end{corollary}
In the subsequent subsections, we present several concrete choices of baseline algorithms $(g_t)_{t\in[T]}$ that enjoy competitive guarantees, together with the corresponding values of $t_0$ and $\kappa$ for each problem.
Note that the maximum cost in each round $B$ is $1$ except for the ski-rental problem.

\subsection{Prophet Inequality}

We first apply our general framework to the repeated prophet inequality $(\Dadv,\prwd;T)$, where the objective is to maximize the expected reward $\prwd(\bx,i)=x_i$ and $\Pi$ is the point-mass distribution on the identity permutation for any $\fD=(D_1,\dots,D_n;\Pi)\in\Dadv$. 
In this setting, we construct a sequence $(g^{\mathrm{PI}}_t)_{t\in[T]}$ of randomized threshold algorithms as follows:
\begin{itemize}
    \item Let $g^{\mathrm{PI}}_1$ be the randomized algorithm that selects one variable uniformly at random.
    \item For $t \ge 2$, let $g^{\mathrm{PI}}_t[\History{t}]$ be the single-sample algorithm for the prophet inequality~\cite{rubinstein_et_al:LIPIcs:2020:11745}, i.e., $g^{\mathrm{PI}}_t$ selects the first variable $X_i^{(t)}$ such that $X_i^{(t)} \ge \max_{i'\in[n]} X_{i'}^{(1)}$.
\end{itemize}
It is easy to see that $g^{\mathrm{PI}}_1$ achieves a competitive ratio of $1/n$.
Rubinstein et al.~\cite{rubinstein_et_al:LIPIcs:2020:11745} proved that the single-sample algorithm $g^{\mathrm{PI}}_t$ ($t \ge 2$) achieves a competitive ratio of $1/2$ for the prophet inequality.
Note that, since the algorithms guarantee competitive ratios for the adversarial order, the same competitive ratios hold even for the general order $\Dall$.

By applying \Cref{thm:main}, we obtain \Cref{cor:generic} for the baseline sequence $(g^{\mathrm{PI}}_t)_{t\in[T]}$ with parameters $t_0=\lfloor\sqrt{nT}\rfloor$ and $\kappa = n$, yielding regret $O(\sqrt{nT\log T})$.
Additionally, the same algorithms with $\kappa = \min(n|\Pi|,en!)$ and $t_0=\sqrt{\kappa T}$ work for the general order prophet inequality.

As mentioned in \Cref{thm:PI-refined},
we can obtain a near-optimal regret bound by leveraging the result of Jin et al.~\cite{jin2024sample}.
See \Cref{subsec:pf-thm:PI-refined} for details.
We remark that the regret bound is improved compared with that derived from \Cref{cor:generic}, while the probability of guaranteeing the baseline performance becomes slightly worse.

\begin{theorem}
\label{thm:PI-refined-restate}
    For the repeated prophet inequality problem $(\Dadv,\prwd;T)$ and any $\fD \in \Dadv$, we can compute a sequence of algorithms $(h^{\mathrm{PI}*}_t)_{t\in[T]}$ that achieves the following:
    \begin{enumerate}
        \item for each $t \le \lfloor\sqrt{T}\rfloor$, $h^{\mathrm{PI}*}_t = g^{\mathrm{PI}}_t$,
        \item for each $t > \lfloor\sqrt{T}\rfloor$, with probability at least $1-2/t$, $h^{\mathrm{PI}*}_t(\fD) \ge g^{\mathrm{PI}}_{\zeta(t)}(\fD)$, and
        \item the regret of $(h^{\mathrm{PI}*}_t)_{t\in[T]}$ is $O(\sqrt{T}\log T)$.
    \end{enumerate}
\end{theorem}

The algorithm used in \Cref{cor:generic} and $h^{\mathrm{PI}*}_t$ differ in the setting of $\epsilon_1$ and $\delta_1$, as well as in how they approximate the optimal online algorithm $\opton$ when the event $\cC_t$ occurs.

\subsection{Adversarial Order Secretary Problem (Best Choice Prophet Inequality)}

Next, we apply our framework to the repeated adversarial order secretary problem $(\Dadv,\pbest;T)$.

We define a sequence $(g^{\mathrm{AS}}_t)_{t\in[T]}$ of randomized threshold algorithms for the adversarial order secretary problem as follows:
\begin{itemize}
    \item Let $g^{\mathrm{AS}}_1$ be the randomized algorithm that selects one variable uniformly at random.
    \item For $2 \le t \le T$, let $g^{\mathrm{AS}}_t$ be the single-sample algorithm for the problem~\cite{DBLP:conf/soda/NutiV23}, i.e., $g^{\mathrm{AS}}_t$ selects the first variable $X_i$ such that $X_i \ge \max_{i'\in[n]} X_{i'}^{(1)}$.
\end{itemize}
For this problem, Nuti and Vondr\'ak~\cite{DBLP:conf/soda/NutiV23} proved that the above simple threshold algorithm achieves a competitive ratio of $1/4$, and no algorithm can have a better competitive ratio in the single sample setting.
We remark that the algorithm also has the same competitive ratio for the repeated general order secretary problem $(\Dall,\pbest;T)$.

Then, we have \Cref{cor:generic} for $(g^{\mathrm{AS}}_t)_{t\in[T]}$ with parameters $t_0=\lfloor \sqrt{nT}\rfloor$ and $\kappa = n$, yielding regret bound $O(\sqrt{nT\log T})$.
Furthermore, the same algorithms with $\kappa = \min(n|\Pi|,en!)$ and $t_0=\sqrt{\kappa T}$ work for the general order secretary problem.

\subsection{Prophet Secretary Problem}

We apply our general framework to the repeated prophet secretary problem $(\Dro,\prwd;T)$.
For $t=1$, let $g^{\mathrm{PS}}_1$ be the classical secretary algorithm that rejects the first $n/e$ variables and accepts the first remaining variable $X_{i}$ ($i > n/e$) such that $X_{i} \ge \max_{1\le i' \le n/e} X_{i'}$. It is well known that this algorithm can accept the maximum realized value $\max_{i'\in[n]}X_{i'}$ with probability at least $1/e$~\cite{ferguson1989solved}. This implies that $g^{\mathrm{PS}}_1$ achieves a competitive ratio of at least $1/e$. 

For $t = 2$, let $g^{\mathrm{PS}}_2$ be the single-sample algorithm for the prophet secretary, i.e., $g^{\mathrm{PS}}_t$ selects the first $X_i$ such that $X_i\ge \max_{i'\in[n]}X^{(1)}_{i'}$.
This algorithm achieves a competitive ratio of $1/2$ even for the prophet secretary problem, because the threshold 
$\max_{i'\in[n]} X^{(1)}_{i'}$ is independent of the ordering of the first round and $1/2$ guarantee holds regardless of the permutation observed in round $t$.

For $t \ge 3$, let $g^{\mathrm{PS}}_t$ be the algorithm with $t-1$ samples presented by Cristi and Ziliotto~\cite{cristi2024prophet}.
They demonstrated that, with $\Omega(1/\epsilon^5)$ samples, we can obtain an $(\alpha^{\rm PS} - \epsilon)$-competitive threshold algorithm, where $\alpha^{\rm PS}$ is the current best known competitive ratio for the prophet secretary problem. $\alpha^{\rm PS}$ is now $0.688$, which was established by~Chen et al.~\cite{Chen2025ProphetSecretary}.
Hence, $g^{\mathrm{PS}}_t$ ($t \ge 3$) achieves a competitive ratio of at least $0.688-O(t^{-\frac15})$.

By \Cref{thm:main}, we have \Cref{cor:generic} for the above-defined $(g^{\mathrm{PS}}_t)_{t\in[T]}$ with $t_0\ge \lfloor\sqrt{n!T}\rfloor$ and $\kappa = n!$, yielding regret $O(\sqrt{n!T\log T})$.

\subsection{IID Prophet Inequality}

Now, we consider the repeated i.i.d.\ prophet inequality $(\Diid,\prwd;T)$.
For this problem, it is known that the best possible competitive ratio is $\alpha^{\rm PIiid}\simeq 0.745$~\cite{hill1982comparisons,kertz1986stop,correa2017posted,liu2020variable}.

For $t=1$, we define $g_1^{\rm PIiid}$ to be the algorithm that rejects the first $\lfloor n/e \rfloor$ variables and accepts the first remaining variable $X_{i}$ ($i > n/e$) such that $X_{i} \ge \max_{1\le i' \le n/e} X_{i'}$. This algorithm achieves a competitive ratio of at least $1/e$, since this setting can be interpreted as a prophet i.i.d.\ secretary problem.
Correa et al.~\cite{correa2019prophet} provided the \emph{fresh-looking samples algorithm}, which is $(1-1/e)$-competitive for the i.i.d.\ prophet inequality when $n-1$ samples $X'_1,\dots,X'_{n-1}$ are available. The algorithm accepts the first variable $X_i$ such that $X_i$ is greater than or equal to the maximum of a random subset of size $n-1$ from the set $\{X'_1,\dots,X'_{n-1},X_1,\dots,X_{i-1}\}$.
This implies that we can construct a $(1-1/e)$-competitive algorithm $g_2^{\rm PIiid}$ for the i.i.d.\ prophet inequality by using the algorithm of Correa et al.~\cite{correa2019prophet}.


In addition, Correa et al.~\cite{correa2024sample} showed that, with $\Omega(1/\epsilon)$ samples from $\fD\in\Diid$, there exists a threshold algorithm for the i.i.d.\ prophet inequality whose competitive ratio is $\alpha^{\rm PIiid}-\epsilon$.
By this result, for $t \ge 3$, we obtain an $(\alpha^{\rm PIiid}-O(t^{-1}))$-competitive algorithm $g_t^{\rm PIiid}$.


Similarly to the proof of \Cref{thm:PI-refined}, we have \Cref{cor:generic} for baseline $(g_t^{\rm PIiid})_{t\in[T]}$ with $\kappa = n$ and $t_0=\lfloor\sqrt{T}\rfloor$.
In this case, the regret bound is $O(\sqrt{T}\log T)$.

\subsection{Random Order Secretary Problem / IID Secretary Problem}

We address the repeated random order secretary problem $(\Dro,\pbest;T)$ and the repeated secretary problem with identical distributions $(\Diid,\pbest;T)$.

For the random order secretary problem, we define baseline algorithms $(g^{\mathrm{RS}}_t)_{t\in[T]}$ as follows:
\begin{itemize}
    \item Let $g^{\mathrm{RS}}_1$ be an algorithm that rejects the first $n/e$ variables and accepts the first remaining variable $X_{i}$ ($i > n/e$) such that $X_{i} \ge \max_{1\le i' \le n/e} X_{i'}$.
    \item For $t \ge 2$, let $g^{\mathrm{RS}}_t$ be the block-rank algorithm~\cite{DBLP:conf/soda/NutiV23} for the random order secretary problem.
\end{itemize}
For the random order secretary problem with a single sample $\bX^{(1)}$, Nuti and Vondr\'ak~\cite{DBLP:conf/soda/NutiV23} showed that the block-rank algorithm achieves a competitive ratio of $0.5009$ for the random order secretary problem and $0.5024-O(1/\log n)$ for the i.i.d.\ secretary problem.
They also established that no algorithm can obtain a competitive ratio better than $0.5024$.


By \Cref{thm:main}, we have \Cref{cor:generic} for $(g^{\mathrm{RS}}_t)_{t\in[T]}$ with $t_0=\lfloor\sqrt{n!T}\rfloor$ and $\kappa = n!$, where the regret bound is $O(\sqrt{n!T\log T})$.
Note that these algorithms are also applicable with $t_0=\lfloor\sqrt{nT}\rfloor$ and $\kappa = n$ to the i.i.d.\ version.

\subsection{Last Success Problem}

We explore the repeated adversarial order last success problem $(\Dadv,\pls;T)$, where the objective is to maximize the probability of selecting the last success index, that is, $\pls(\bx,i)=\Ind{i=\max\{i'\mid x_{i'}=1\}}$.

Let $g^{\mathrm{LS}}_1$ be the randomized algorithm that selects one variable uniformly at random.
For $2 \le t \le T$, let $g^{\mathrm{LS}}_t$ be the one-sample algorithm for the adversarial order last success problem that accepts the first success after the index of the second-to-last success in the sample sequence.

Yoshinaga and Kawase~\cite{yoshinaga2024last} proved that the algorithm accepts the last success with probability at least $1/4$ if $R\ge (\sqrt{3}-1)/2$ and with probability $\frac{R(4+3R)}{4(1+R)^2}$ if $0\le R\le (\sqrt{3}-1)/2$, where
$\textstyle R=\sum_{i=1}^n \Prob{X_i=1}/\Prob{X_i<1}.$
Combining this with the fact that 
$$
\optoff(\fD;\pls)=1-\prod_{i=1}^n\Prob{X_i<1}\le 1-e^{-R},$$ 
we can obtain that the competitive ratio of the algorithm is at least $1/4$.

By \Cref{thm:main}, we have \Cref{cor:generic} for $(g^{\mathrm{LS}}_t)_{t\in[T]}$ with
$t_0=\lfloor \sqrt{nT}\rfloor$ and $\kappa = n$.
Note that the one-sample algorithm is also applicable in the general order model $\Dall$, since we can count the number of successes in the sample sequence from the distributions $D_i$ that have not yet been observed.

\subsection{Ski-rental Problem}

We analyze the repeated ski-rental problem $(\Dadv,\pski;T)$, where the objective is to minimize the total cost $p(\bx,i)=\sum_{j=1}^{i-1}x_j+b$. Note that the maximum cost $B$ in each round is at most $n+b$.

It is known that there exists a randomized algorithm without samples that achieves a competitive ratio of $e/(e-1)$ for this problem~\cite{karlin1994competitive}.\footnote{This online algorithm works in our setting by viewing each rental cost $x_i$ as a sequence of infinitesimally small costs.}
This online algorithm accepts the $i$th variable if the cumulative rental cost $\sum_{j=1}^{i}x_j$ exceeds a randomly preselected value.
Then, for each $t\in[T]$, let $g^{\rm SKI}_t$ be the above algorithm without samples.
By \Cref{thm:main}, we have \Cref{cor:generic} for $(g^{\rm SKI}_t)_{t\in [T] }$ with $t_0=\lfloor\sqrt{nT}\rfloor$ and $\kappa=n$.
In this problem, the regret bound is $O((n+b)\sqrt{nT\log T})$.

\section{Standard Lemma}

\begin{lemma} \label{lem:kl-bound}
For any constant $0 < c < 1/2$, there exists a constant $C_c > 0$ such that for all $x, y \in [-c, c]$, 
\[
\KLdiv{\mathrm{Ber}(1/2 + x)}{\mathrm{Ber}(1/2 + y)} \le C_c (x - y)^2.
\]
\end{lemma}

\begin{proof}
Using the standard inequality between KL divergence and $\chi^2$-divergence, we have:
\begin{align}
\KLdiv{\mathrm{Ber}(1/2 + x)}{\mathrm{Ber}(1/2 + y)}
&\le \frac{(x-y)^2}{1/2 - y} + \frac{(x-y)^2}{1/2 + y}
= \frac{(x - y)^2}{1/4 - y^2}.
\end{align}
Since $|y| \le c$, the denominator is lower-bounded by $1/4 - c^2 > 0$. Thus, the inequality holds by setting $C_c = (1/4 - c^2)^{-1}$.
\end{proof}

\section{Proof of \Cref{lem:threshold}}\label{subsec:proof threshold}
We construct a threshold algorithm $h^*=(f^*_i)_{i=1}^n$ iteratively from $i=n$ to $1$.

For $i\in[n]$, $\bx_{\le i}\in [0,1]^i$, and $\tau_{\le i}\in\Inj([i],[n])$, let $\beta(\bx_{\le i},\tau_{\le i})$ be the optimal expected profit (of an online algorithm) if the realization of $(\bX_{\le i}, \pi_{\le i})$ is $(\bx_{\le i},\tau_{\le i})$, and the first $i$ variables are rejected, that is,
\begin{align}
\beta(\bx_{\le i},\tau_{\le i})
&= \sup_{h\in\cH:\,\text{always reject first $i$ variables}}\cExpec{h(\bX,\pi;p)}{\begin{array}{l}\bX_{\le i}=\bx_{\le i},\\\pi_{\le i}=\tau_{\le i}\end{array}}\\
&= \sup_{f_{i+1},\dots,f_n}\cExpec{p(\bX,\min\{j\in \{i+1,\dots,n\}\mid f_j(\bx_{\le j},\tau_{\le j})=1\})}{\begin{array}{l}\bX_{\le i}=\bx_{\le i},\\\pi_{\le i}=\tau_{\le i}\end{array}},
\end{align}
where we denote $\min \emptyset = n+1$.
By definition, we have $\beta(\bx,\tau)=p(\bx,n+1)$ for any $\bx\in [0,1]^n$ and $\tau \in \Inj([n],[n])$.
In addition, for $i\in[n]$, $\bx_{\le i}\in [0,1]^i$, and $\tau_{\le i}\in\Inj([i],[n])$, let $\alpha(\bx_{\le i},\tau_{\le i})$ be the expected profit if the realization of $(\bX_{\le i}, \pi_{\le i})$ is $(\bx_{\le i}, \tau_{\le i})$, and the $i$th variable is accepted, i.e., 
\begin{align}
\alpha(\bx_{\le i},\tau_{\le i})\coloneqq
\cExpec{p(\bX,i)}{\bX_{\le i}=\bx_{\le i}\text{ and }\pi_{\le i}=\tau_{\le i}}.
\end{align}
Then, when the realizations of the first steps are $\bx_{\le i}$ and $\tau_{\le i}$ and the first $i-1$ variables are rejected,
an optimal online algorithm accepts the $i$th variable if 
\begin{align}
\alpha(\bx_{\le i},\tau_{\le i}) \ge \beta(\bx_{\le i},\tau_{\le i}).
\end{align}

Suppose that 
\begin{align}
(f^*_{i+1},\dots,f^*_n)\in\argmax_{f_{i+1},\dots,f_n}\cExpec{p(\bX,\min\{j\in \{i+1,\dots,n\}\mid f_j(\bx_{\le j},\tau_{\le j})=1\})}{\begin{array}{l}\bX_{\le i}=\bx_{\le i},\\\pi_{\le i}=\tau_{\le i}\end{array}}.
\end{align}
For notational convenience, let $\gtrsim_j$ represent $\ge$ if $f^*_j(\bx_{\le j},\tau_{\le j}) = \Ind{x_j \ge \theta_j(\tau_{\le j})}$ and $>$ otherwise, i.e., $f^*_j(\bx_{\le j},\tau_{\le j}) = \Ind{x_j > \theta_j(\tau_{\le j})}$.
Then, we have
\begin{align}
\MoveEqLeft[1]
\beta(\bx_{\le i},\tau_{\le i})-\alpha(\bx_{\le i},\tau_{\le i})\\
&=\max_{f_{i+1},\dots,f_n}\cExpec{p\big(\bX,\min\big\{j\in \{i+1,\dots,n\}\mid f_j(\bx_{\le j},\tau_{\le j})=1\big\}\big)-p(\bX,i)}{\begin{array}{l}\bX_{\le i}=\bx_{\le i},\\\pi_{\le i}=\tau_{\le i}\end{array}}\\
&=\cExpec{p\big(\bX,\min\big\{j\in \{i+1,\dots,n\}\mid f^*_j(\bx_{\le j},\tau_{\le j})=1\big\}\big)-p(\bX,i)}{\begin{array}{l}\bX_{\le i}=\bx_{\le i},\\\pi_{\le i}=\tau_{\le i}\end{array}}\\
&=\cExpec{p\big(\bX,\min\big\{j\in \{i+1,\dots,n\}\mid {X_j\gtrsim_j \theta_j(\tau_{\le j})}\big\}\big)-p(\bX,i)}{\begin{array}{l}\bX_{\le i}=\bx_{\le i},\\\pi_{\le i}=\tau_{\le i}\end{array}}\\
&=\cExpec{q_i\big(\bX_{\ge i},\min\big\{j\in \{i+1,\dots,n\}\mid {X_j\gtrsim_j \theta_j(\tau_{\le j})}\big\}\big)}{\begin{array}{l}\bX_{\le i}=\bx_{\le i},\\\pi_{\le i}=\tau_{\le i}\end{array}}\\
&=\cExpec{q_i\big(\bX_{\ge i},\min\big\{j\in \{i+1,\dots,n\}\mid {X_j\gtrsim_j \theta_j(\tau_{\le j})}\big\}\big)}{\begin{array}{l}\bX_{i}=\bx_{i},\\\pi_{\le i}=\tau_{\le i}\end{array}}. \label{eq:gamma-pstar}
\end{align}
Thus, whether $\alpha(\bx_{\le i},\tau_{\le i}) \ge \beta(\bx_{\le i},\tau_{\le i})$ does not depend on the realization of the first $i-1$ variables.

Define a threshold as
\begin{align}
    \theta_i(\tau_{\le i}) \coloneqq \inf\left\{x_i\in[0,1] ~\middle|~ \cExpec{q_i\big(\bX_{\ge i},\min\big\{j\in \{i+1,\dots,n\}\mid {X_j\gtrsim_j \theta_j(\tau_{\le j})}\big\}\big)\!}{\!\!\begin{array}{l}X_i=x_i,\\\pi_{\le i}=\tau_{\le i}\end{array}\!\!}\le 0\right\}.
\end{align}
If $\alpha((\bx_{\le i-1},x),\tau_{\le i}) < \beta((\bx_{\le i-1},x),\tau_{\le i})$ for all $x\in[0,1]$, then we set $\theta_i(\tau_{\le i}) = \infty$.
Note that the definition of $\theta_i$ does not depend on the choice of $\bx_{\le i-1}$.

By the properties (\ref{prop:p-i}) and (\ref{prop:p-ii}), the value of $q_i(\bx_{\ge i},j)=p(\bx,j)-p(\bx,i)$ is monotone non-increasing with respect to $x_i$. 
This implies that 
$$
\cExpec{q_i\big(\bX_{\ge i},\min\big\{j\in \{i+1,\dots,n\}\mid {X_j\gtrsim_j \theta_j(\tau_{\le j})}\big\}\big)}{X_i=x_i\text{ and }\pi_{\le i}=\tau_{\le i}}
$$ is also monotone non-increasing with respect to $x_i$.
Therefore, if $x_i>\theta_i(\tau_{\le i})$, we have $\alpha(\bx_{\le i},\tau_{\le i})\ge\beta(\bx_{\le i},\tau_{\le i})$ by \eqref{eq:gamma-pstar}.
Also, if $x_i<\theta_i(\tau_{\le i})$, we have $\alpha(\bx_{\le i},\tau_{\le i})<\beta(\bx_{\le i},\tau_{\le i})$ by \eqref{eq:gamma-pstar}.
Thus, $f_i^*(\bx_{\le i},\tau_{\le i})=\Ind{x_i\ge\theta_i(\tau_{\le i})}$ or $\Ind{x_i >\theta_i(\tau_{\le i})}$ is an optimal choice of the function.
Hence, by induction, we conclude that the threshold algorithm $h^*=(f_i^*)_{i=1}^n$ is an optimal online algorithm.

\section{Proof of \Cref{thm:semibandit-impossible}}
\label{app:pf-semibandit-impossible}

Let $T \ge e^4$, and consider a prophet inequality instance with $n = 2$, where $X_1 \equiv 1/2$ and $X_2 \sim \mathrm{Ber}(q)$.
Let $\fD_+$ denote the distribution with $q = 1/2 + 1/\log T$, and $\fD_-$ denote the distribution with $q = 1/2 - 1/\log T$.
Let the baseline algorithm $g_t$ always select $X_1$.

Fix an arbitrary algorithm sequence $(h_t)_{t \in [T]}$ that satisfies the baseline condition. That is, assume that there exists a constant $C > 0$ such that for all $\fD \in \DD$ and $t \in [T]$, the probability that the expected reward of $h_t$ is at least that of $g_t$ is at least $1 - C/t$.
We show that the regret with respect to $\fD_+$ is $\Omega(T / \log T)$.

Let $A_t$ denote the index of the variable selected by $h_t$, and let $N_i$ $(i = 1,2)$ denote the number of times $X_i$ is selected over all rounds.
Let $\Prob[\pm]{\cdot}$ denote the probability measure over the observation history under $\fD_{\pm}$, and let $\Expec[\pm]{\cdot}$ denote the expectation with respect to the joint distribution of $(A_t)_{t \in [T]}$ under $\fD_{\pm}$.

Under $\fD_-$, the expected reward of selecting $X_2$ is $1/2 - 1/\log T$, which is strictly smaller than the reward $1/2$ of $g_t$, and this holds regardless of the past history. Therefore, the event that the conditional expected reward of $h_t$, given the past samples, is smaller than that of $g_t$ is equivalent to the event $A_t = 2$.
Hence, the probability that $h_t$ selects (and observes) $X_2$ at round $t$ is at most $C/t$, and thus
$\Expec[-]{N_2} \le \sum_{t=1}^{T} {C}/{t} = O(\log T)$.
By the definition of total variation distance $d_{\mathrm{TV}}(\cdot,\cdot)$ and Pinsker's inequality, we have
\[
\left|\Prob[+]{A_t = 2} - \Prob[-]{A_t = 2}\right|
\le d_{\mathrm{TV}}(\Pr_+, \Pr_-)
\le \sqrt{\KLdiv{\Pr_-}{\Pr_+}/2}.
\]

By the chain rule of KL divergence and the fact that the distributions differ only when $A_t = 2$, we obtain
\begin{align}
  \KLdiv{\Pr_-}{\Pr_+} 
  &= \Expec[-]{N_2} \cdot \KLdiv*{\mathrm{Ber}\left(\frac12 - \frac{1}{\log T}\right)}{\mathrm{Ber}\left(\frac12 + \frac{1}{\log T}\right)} \\
  & \le O(\log T) \cdot O( \log^{-2} T )
  = O(\log^{-1} T),
\end{align}
where the inequality follows from $\log T \ge 4$ and \Cref{lem:kl-bound}.

Therefore, $\left|\Prob[+]{A_t = 2} - \Prob[-]{A_t = 2}\right| = O(\log^{-1/2} T)$, which implies
\begin{align}
  \Expec[+]{N_2}
  &\le \Expec[-]{N_2} + T \cdot O(\log^{-1/2} T)
  = O(T \cdot \log^{-1/2} T),\\
  \Expec[+]{N_1}
  &= T - \Expec[+]{N_2}
  = T - O(T \cdot \log^{-1/2} T)
  = \Omega(T).
\end{align}

Under $\fD_+$, each time $X_1$ is selected incurs a regret of $1/\log T$, and therefore the regret of $(h_t)_{t \in [T]}$ under $\fD_+$ is $\Omega(T/\log T)$.

\section{Proof of \Cref{thm:full-feedback-as-impossible}}\label{sec:impossible}

Consider a two-variable prophet inequality problem with random variables $X_1 \equiv 1/2$ and $X_2 \sim \mathrm{Ber}(q)$. Let the baseline policy $g_t$ always select $X_1$.

We define two distributions $\fD_+$ and $\fD_-$ by setting $q = 1/2 + \epsilon$ and $q = 1/2 - \epsilon$, respectively, for some fixed constant $\epsilon \in (0,1/2)$. Note that under $\fD_+$, the optimal policy selects $X_2$, while under $\fD_-$, the baseline policy $g_t$ is optimal.

Consider any sequence of algorithms $(h_t)_{t \in [T]}$. In the full-feedback setting, at each round the algorithm observes the realizations of both $X_1$ and $X_2$, and hence after $t-1$ rounds, its decision at round $t$ can be written as a function of the past observations of $X_2$. In particular, let $m$ denote the number of past rounds in which $X_2 = 1$.

Fix any round $t$.
For any $m \in \{0,\dots,t-1\}$, the event that $X_2=1$ occurs exactly $m$ times in the first $t-1$ rounds has positive probability under both $\fD_+$ and $\fD_-$.
Therefore, conditioned on any fixed value of $m \in \{0,1,\dots,t-1\}$, if the algorithm selects $X_2$ with positive probability, then under $\fD_-$ the expected reward at round $t$ is strictly smaller than the baseline value $1/2$, since $\mathbb{E}[X_2] = 1/2 - \epsilon$.


It follows that, in order to satisfy
$\Prob{h_t(\fD_-) \ge g_t(\fD_-)} = 1$
for every round $t \in [T]$, the algorithm must select $X_1$ with probability one at every round, regardless of the observed history.
However, under $\fD_+$, this strategy incurs regret $\epsilon$ at each round, and hence the total regret is $\epsilon T = \Omega(T)$.

Therefore, no algorithm sequence $(h_t)_{t \in [T]}$ can simultaneously achieve sublinear regret on $\fD_+$ and guarantee
$\Prob{h_t(\fD_-) \ge g_t(\fD_-)} = 1$
for all $t \in [T]$, completing the proof.




\section{Proof of \Cref{prop:naive-failure}}
\label{app:pf-prop:naive-failure}

Let $T \ge 64$ and
consider a prophet inequality instance with $n = 2$, where $X_1 \equiv 1/2$ and $X_2 \sim \mathrm{Ber}(q)$.
Let $\fD_+$ and $\fD_-$ be the distributions with $q = 1$ and $q = 1/2 - 1/T$ respectively.
At the beginning of round $t$, let $N_2^{(t)}$ be the number of times $X_2^{(s)} = 1$ for $s \in [t-1]$, and define the unbiased estimator $\hat{q}^{(t)} \coloneqq N_2^{(t)}/(t-1)$.

At round $t$, let $h_t$ be an empirically optimal algorithm that selects $X_2$ if $\hat{q}^{(t)} \ge 1/2$, and selects $X_1$ otherwise. Let the baseline algorithm $g_t$ always select $X_1$ and obtain reward $1/2$.

We check the performance of the algorithm for $\fD_+$.
In this case, selecting $X_2$ is the best option, and the algorithm incurs regret of at least $\sum_{t \in [T]} \epsilon_t/2$.
If $\sum_{t \in [T]} \epsilon_t/2 \geq T/8$, then the algorithm incurs a linear regret, and we are done.

Suppose $\sum_{t \in [T]} \epsilon_t/2 < T / 8$.
We can observe that there exist at least $T/2$ rounds $t \in [T]$ such that $\epsilon_t < 1/2$.
%
For $\fD_-$, we have $h_t(\fD_-) < g_t(\fD_-)$ if and only if $h_t$ selects $X_2$ with positive probability. 
We will show that $h_t$ incorrectly selects $X_2$ with probability at least $1/8$ for any round $t$ with $\epsilon_t < 1/2$.

The empirically optimal algorithm $h_{t}$ selects $X_2$ if and only if $\hat{q}^{(t)} \ge 1/2$. When $q = 1/2$, by symmetry, the probability that $\hat{q}^{(t)} \ge 1/2$ is at least $1/2$.
Let $D_0^{\otimes}$ be the joint distribution of $t-1$ i.i.d.\ samples from $\mathrm{Ber}(1/2)$, and let $D_-^{\otimes}$ be that from $\mathrm{Ber}(1/2 - 1/T)$. Then, by Pinsker's inequality, the total variation distance satisfies
\begin{align}
  d_{\mathrm{TV}}(D_0^{\otimes}, D_-^{\otimes})
  &\le \sqrt{\frac{t - 1}{2} \cdot \mathrm{KL}\Big(\mathrm{Ber}(1/2 - 1/T) \;\Big\|\; \mathrm{Ber}(1/2)\Big)} \\
  & \le 2\sqrt{T \cdot 1/T^2} = 2/\sqrt{T} \leq 1/4.
\end{align}
where the second inequality is due to $T \ge 64$ and \Cref{lem:kl-bound} with $c=1/64$ and $C_c \leq 8$.
Therefore, in the actual instance, the probability that $h_{t}$ incorrectly selects $X_2$ is
\[
\Prob[q=1/2-1/T]{\hat{q}^{(t)} \ge 1/2}
\ge
\Prob[q=1/2]{\hat{q}^{(t)} \ge 1/2} - d_{\mathrm{TV}}(D_0^{\otimes}, D_-^{\otimes})
\ge
\frac{1}{2} - \frac{1}{4}=\frac{1}{4}.
\]
Therefore, when $\epsilon_t < 1/2$,
we have
\[
\Prob{h'_t(\fD_-) < g_t(\fD_-)}
= (1 - \epsilon_t)\cdot\Prob[q=1/2-1/T]{\hat{q}^{(t)} \ge 1/2} \ge 1/8.
\]


\section{Proof of \Cref{lem:massart}}\label{subsec:massart}
To prove \Cref{lem:massart}, we use the following lemma.
\begin{lemma}[Massart's Lemma~{\cite[Lemma 5.2]{massart2000some}}]
\label{lem:massart---}
    Let $A$ be a finite subset of $\R^t$, and let $r \coloneqq \max_{a \in A} \|a\|_2$. Then,
    \begin{align}
        \E_{\sigma} \left[ \frac{1}{t} \sup_{a \in A} \sum_{s=1}^t \sigma_s a_s \right] \le \frac{r \sqrt{2 \ln |A|}}{t},
    \end{align}
    where $\sigma_1,\ldots, \sigma_t$ are independent Rademacher random variables.
\end{lemma}

Let $-A \coloneqq \{-a \mid a \in A\}$. Obviously, the following holds:
\begin{align}
    \sup_{a \in A \cup -A} \sum_{s=1}^t \sigma_s a_s
    =  \sup_{a \in A} \left| \sum_{s=1}^t \sigma_s a_s \right|.
\end{align}
Then, we have
\begin{align}
    \rade(A) 
    = \E_{\sigma} \left[ \frac{1}{t} \sup_{a \in A}\left| \sum_{s=1}^t \sigma_s a_s \right| \right] 
    = \E_{\sigma} \left[ \frac{1}{t} \sup_{a \in A \cup -A} \sum_{s=1}^t \sigma_s a_s \right] 
    \le \frac{r \sqrt{2 \ln (2|A|)}}{t},
\end{align}
where the inequality is due to Lemma~\ref{lem:massart---}.

\section{Omitted Proofs from \Cref{sec:general}}
\label{sec:omitted-proofs-sec4}

\subsection{Proof of \Cref{thm:uniform-comp-and-regret}}
\label{subsec:pf-thm:uniform-comp-and-regret}

    First, we show that $h_t^*(\fD) \ge g_{\zeta(t)}(\fD)$ when $\cE_t$ occurs.
    Under the event $\cE_t \wedge \cC_t$, $h_t^*$ uses $h_{\zeta(t)}$ and we have
    \begin{align}
        h_{\zeta(t)}(\fD)
        &\ge \hat{h}_{\zeta(t)}(\History{t};\zeta(t)) - \epsilon(t) \\
        &\ge \hat{g}_{\zeta(t)}(\History{t};\zeta(t)) + \epsilon(t)
        \ge g_{\zeta(t)}(\fD).
    \end{align}
    Then, we obtain $h_t^*(\fD) \ge g_{\zeta(t)}(\fD)$.
    Under the event $\cE_t \wedge \neg\cC_t$, $h_t^*$ uses $g_{\zeta(t)}$, and then we also obtain $h_t^*(\fD) \ge g_{\zeta(t)}(\fD)$.
    Therefore, when $\cE_t$ occurs, we have $h_t^*(\fD) \ge g_{\zeta(t)}(\fD)$.

    Next, we bound the expected regret. 
    We have
    \begin{align}
        \MoveEqLeft[1]
        \opton(\fD) - \Expec{h_t^*(\fD)}\\
        &= \Prob{\cE_t}\cdot\cExpec{\opton(\fD)-h_t^*(\fD)}{\cE_t} + \Prob{\neg\cE_t}\cdot\cExpec{\opton(\fD)-h_t^*(\fD)}{\neg\cE_t} \\
        & \le \cExpec{\opton(\fD)-h_t^*(\fD)}{\cE_t} + B\delta(t) \\
        & = \begin{multlined}[t]
            \Prob{\cC_t\mid\cE_t}\cdot\cExpec{\opton(\fD)-h_{\zeta(t)}(\fD)}{\cC_t,\cE_t} \\
            + \Prob{\neg\cC_t\mid\cE_t}\cdot\cExpec{\opton(\fD)-g_{\zeta(t)}(\fD)}{\neg\cC_t,\cE_t} + B\delta(t)   
        \end{multlined}
        \\
        & = \begin{multlined}[t]
        \Prob{\cC_t\mid\cE_t}\cdot\cExpec{\opton(\fD)-h_{\zeta(t)}(\fD)}{\cC_t,\cE_t} \\
        +\Prob{\neg\cC_t\mid\cE_t}\cdot\cExpec{\opton(\fD)-h_{\zeta(t)}(\fD)+h_{\zeta(t)}(\fD)-g_{\zeta(t)}(\fD)}{\neg\cC_t,\cE_t} + B\delta(t) 
        \end{multlined}\\
        & = \cExpec{\opton(\fD)-h_{\zeta(t)}(\fD)}{\cE_t} + \Prob{\neg\cC_t\mid\cE_t}\cdot\cExpec{h_{\zeta(t)}(\fD)-g_{\zeta(t)}(\fD)}{\neg\cC_t,\cE_t} + B\delta(t) \\
        & \le \Prob{\neg\cC_t\mid\cE_t}\cdot\cExpec{h_{\zeta(t)}(\fD)-g_{\zeta(t)}(\fD)}{\neg\cC_t,\cE_t}
        + 2\epsilon(t) + B\delta(t),
    \end{align}
    where the first inequality holds because $\Prob{\neg\cE_t}\le \delta(t)$ and $\opton(\fD)\le B$, and $\Prob{\neg\cE_t}\le \delta(t)$, and the second inequality is due to the definition of $\cE_t$.

    For readability, we abbreviate $\hat{g}_{\zeta(t)}(\History{t};\zeta(t))$ and $\hat{h}_{\zeta(t)}(\History{t};\zeta(t))$ as $\hat{g}_{\zeta(t)}$ and $\hat{h}_{\zeta(t)}$ respectively.
    Then, the term $\cExpec{h_{\zeta(t)}(\fD) - g_{\zeta(t)}(\fD)}{\neg\cC_t,\cE_t}$ is bounded as
    \begin{align}
        \MoveEqLeft
        \cExpec{h_{\zeta(t)}(\fD) - g_{\zeta(t)}(\fD)}{\neg\cC_t,\cE_t}\\
        & \le \cExpec{\hat{h}_{\zeta(t)} - \hat{g}_{\zeta(t)}}{\neg\cC_t,\cE_t}
        + 2\epsilon(t) \\
        & = \cExpec{(\hat{h}_{\zeta(t)}-\epsilon(t)) - (\hat{g}_{\zeta(t)} + \epsilon(t))}{\neg\cC_t,\cE_t} + 4\epsilon(t) \le 4\epsilon(t).
    \end{align}
    By combining these, we obtain
    \begin{align}
        \opton(\fD) - \Expec{h_t^*(\fD)} 
        \le (2\epsilon(t) + B\delta(t)) + 4\epsilon(t)
        = 6\epsilon(t) + B\delta(t).  
    \end{align}    
This completes the proof.

\subsection{Proof of \Cref{col:delta-epsilon}}
\label{app:pf-col:delta-epsilon}

For any positive integer $t$, let $M_{t-1}$ be a positive integer such that
\begin{align}\label{eq:bound-M}
M_{t-1}= \sup_{(\bx^{(1)},\tau^{(1)}),\dots,(\bx^{(t-1)},\tau^{(t-1)})}\left|\left\{\big(h(\bx^{(1)},\tau^{(1)}),\dots,h(\bx^{(t-1)},\tau^{(t-1)})\big) ~\middle|~ h\in\cHth \right\}\right|
\end{align}
where the supremum is taken over all choices of $\bx^{(1)},\dots,\bx^{(t-1)}$ in $[0,1]$ and $\tau^{(1)},\dots,\tau^{(t-1)}$ in the set of permutations that occur with positive probability.
By utilizing the uniform concentration bound (\Cref{thm:ulln}), we can select $\epsilon_1(t)$ and $\delta_1(t)$ such that the assumption \eqref{eq:sample-complexity} holds as follows.
\begin{claim}\label{lem:delta-epsilon}
    For each integer $t\ge 2$ and any $M\ge M_{t-1}$, the algorithm $h_{\hat{\fD}^{(t)}}$ and the following settings of $\epsilon_1(t)$ and $\delta_1(t)$ satisfy the assumption \eqref{eq:sample-complexity}:
    \begin{align}
        \epsilon_1(t) \coloneqq 6B \sqrt{{2 \ln (2M)}/{(t-1)}} \text{ and } \;
        \delta_1(t) &\coloneqq {1}/{M}.
    \end{align}
\end{claim}
\begin{claimproof}
    Since we assume that the range of the profit is bounded by $[0,B]$, we have
    \begin{align}           \left\|\big(h(\bx^{(1)},\tau^{(1)}),\dots,h(\bx^{(t-1)},\tau^{(t-1)})\big)\right\|_2\le B\sqrt{t-1}
    \end{align}
    for any $(\bx^{(1)},\tau^{(1)}),\dots,(\bx^{(t-1)},\tau^{(t-1)})$ in the support of $\fD$.
    Setting $\eta = B\sqrt{\frac{2 \ln (2M)}{t-1}}$, we have
    \begin{align}
    2\rade_{t-1,\fD}(\cHth)+\eta 
    &\le 2\frac{B\sqrt{t-1}\cdot\sqrt{2\ln(2M)}}{t-1}+ B\sqrt{\frac{2 \ln (2M)}{t-1}}
    = 3B\sqrt{\frac{2 \ln (2M)}{t-1}}=\frac{\epsilon_1(t)}{2},
    \end{align}
    where we use \Cref{lem:massart} to obtain the inequality.
    Then, by \Cref{thm:ulln}, we have
    \begin{align}
    \MoveEqLeft
        \Prob[(\bX^{(1)},\pi^{(1)}),\dots,(\bX^{(t-1)},\pi^{(t-1)}) \simiid \fD]{\sup_{h\in\cHth}\left|h(\fD) - h(\hat{\fD}^{(t)})\right|>\frac{1}{2}\epsilon_1(t)}\\
        &\le \Prob[(\bX^{(1)},\pi^{(1)}),\dots,(\bX^{(t-1)},\pi^{(t-1)}) \simiid \fD]{\sup_{h\in\cHth}\left|h(\fD) - h(\hat{\fD}^{(t)})\right|> 2\rade_{t-1,\fD}(\cHth)+\eta}\\
        &\le 2\exp\left(-\frac{(t-1)\eta^2}{2B^2}\right)
        = 2\exp\left(-\frac{(t-1)\cdot (B^2\cdot 2\ln(2M)/(t-1))}{2\cdot B^2}\right)\\
        &= 2\exp\left(-\ln(2M)\right)
        =\frac{1}{M}=\delta_1(t).
    \end{align}
    Thus, with probability at least $1-\delta_1(t)$, we obtain $\sup_{h\in\cHth}\left|h(\fD) - h(\hat{\fD}^{(t)})\right| \leq \frac{1}{2}\epsilon_1(t)$ for each $h \in \cHth$. If this holds, then
    \begin{align}
    h_{\hat{\fD}^{(t)}}(\fD) 
    &\ge h_{\hat{\fD}^{(t)}}(\hat{\fD}^{(t)}) - \epsilon_1(t)/2
    \ge h_{\fD}(\hat{\fD}^{(t)}) - \epsilon_1(t)/2\\
    &\ge h_{\fD}(\fD) - \epsilon_1(t)
    = \opton(\fD) - \epsilon_1(t),
    \label{eq:opt-approximate-}
    \end{align}
    where we use the optimality of $h_{\hat{\fD}^{(t)}}(\hat{\fD}^{(t)})$ in the second inequality.
    Therefore, $h_{\hat{\fD}^{(t)}}$ and the above-defined $\epsilon_1(t),\delta_1(t)$ satisfy the assumption \eqref{eq:sample-complexity}.
\end{claimproof}

In general, $M_{t-1}$ in \eqref{eq:bound-M} can be upper bounded by $(n+1)^{t-1}$, since each of the $t-1$ rounds can have $n+1$ different outcomes.
However, combining this upper bound with \Cref{lem:delta-epsilon} only yields a regret bound that is linear in the number of rounds $T$.
Thus, we need a more refined upper bound for $M_{t-1}$ to achieve a sublinear regret bound, but in general this is not possible.
To handle this, we focus on a specific class of algorithms $\cHth$.

\begin{claim}
\label{lem:general-finite}
    Suppose that $\cHth$ is the set of threshold algorithms or the set of essentially threshold algorithms.
    For each $t\in[T]$ and any constant $C\ge 1$, we have
    \begin{align}
    M_{t-1}\le t^{\sum_{i=1}^n |\{\tau_{\le i}^{(s)}\mid s\in[t-1]\}|} \le t^{\min(n|\Pi|,en!)} \le (Ct)^{\min(n|\Pi|,en!)},
    \end{align}
    where $|\Pi|$ be the number of distinct permutations that appear with positive probability in $\fD$.
\end{claim}

\begin{claimproof}
    Recall that any threshold algorithm $h \in \cHth$ can be specified by a sequence of threshold functions $(\theta_1, \ldots, \theta_n)$, where $\theta_i\colon \Inj([i],[n])\to[0,1]\cup\{\infty\}$ for each $i\in[n]$.
    For each $i \in [n]$ and $\tau_{\le i}\in\Inj([i],[n])$, consider the set of observed values denoted by $\chi(\tau_{\le i})=\{x_i^{(s)} \mid \tau^{(s)}_{\le i}=\tau_{\le i}, s\in [t-1]\}$. 
    The threshold $\theta_i(\tau_{\le i})$ can be classified into at most $|\chi(\tau_{\le i})|+1$ types, depending on which of the $|\chi(\tau_{\le i})|+1$ subintervals of $[0,\infty]$, determined by partitioning it with the observed values, the threshold falls into.

    Therefore, in round $t$, the total number of distinct output patterns over all $n$ steps is at most 
    \begin{align}
    \prod_{i=1}^n \prod_{\tau_{\le i}\in\Inj([i],[n])}(|\chi(\tau_{\le i})|+1)
    &=\prod_{i=1}^n \prod_{\tau_{\le i}\in\{\tau_{\le i}^{(s)}\mid s\in[t-1]\}}(|\chi(\tau_{\le i})|+1)\\
    &\le \prod_{i=1}^n \prod_{\tau_{\le i}\in\{\tau_{\le i}^{(s)}\mid s\in[t-1]\}}t\\
    &= t^{\sum_{i=1}^n |\{\tau_{\le i}^{(s)}\mid s\in[t-1]\}|}.
    \end{align}
    Hence, $M_{t-1}\le t^{\sum_{i=1}^n |\{\tau_{\le i}^{(s)}\mid s\in[t-1]\}|}$ if $\cHth$ is the set of threshold algorithms.

    If $\cHth$ is the set of essentially threshold algorithms, we can similarly show that the number of distinct output patterns is at most $t^{\sum_{i=1}^n |\{\tau_{\le i}^{(s)}\mid s\in[t-1]\}|}$ since an essentially threshold algorithm is uniquely determined by a threshold algorithm.
    If $|\Pi|$ is $n!$ (e.g., the random order model and the general order model), we have $\sum_{i=1}^n |\{\tau_{\le i}^{(s)}\mid s\in[t-1]\}|\le \sum_{i=1}^n n!/i!\le en!$.
    In addition, the number $\max_{i\in[n]}|\{\tau_{\le i}^{(s)}\mid s\in[t-1]\}|=|\{\tau^{(s)}\mid s\in[t-1]\}|$ is at most $|\Pi|$ and we obtain $\sum_{i=1}^n |\{\tau_{\le i}^{(s)}\mid s\in[t-1]\}|\le n |\Pi|$.
    This completes the proof.
\end{claimproof}

By combining \Cref{lem:delta-epsilon,lem:general-finite} with $C=2$, we can set $\epsilon_1(t)$ and $\delta_1(t)$ to satisfy the assumption~\eqref{eq:sample-complexity}.

\subsubsection{Refined Analysis for the Random Order Model}
\label{subsubsec:refined-random-order}

For the random order model, the threshold $\theta_i$ depends only on the set $\{\tau(1),\dots,\tau(i-1)\}$ and $\tau(i)$. Therefore, we obtain a better bound of $M_{t-1}\le t^{n\cdot 2^{n-1}}$, since $\sum_{i=1}^n n\cdot\binom{n-1}{i-1}=n\cdot 2^{n-1}$, which yields the bound $O(B\sqrt{n 2^{n-1} T \log T})$.

\subsection{Proof of \Cref{thm:main}}
\label{subsec:pf-main}
    Let $\kappa \coloneqq \min(n|\Pi|,en!)$ for short.
    By \Cref{thm:uniform-comp-and-regret} and \Cref{col:delta-epsilon}, we have
    \begin{align}
    h_t^*(\fD)\ge g_{\zeta(t)}[\History{\zeta(t)}](\fD)
    \end{align}
    with probability at least $1-\delta(t)$.
    We upper bound $\epsilon(t)$ and $\delta(t)$ used in the definition of $\cE_t$ and $\cC_t$ when $t > t_0$.
    Since $\zeta(t)\ge t/2$ and $\ln(4x)/(x-1)$ is monotone decreasing for $x>1$, we have
    \begin{align}
        \epsilon(t)
        =\epsilon_1(\zeta(t))
        \le \epsilon_1(t/2)
        =6B \sqrt{\frac{2\kappa \ln (2t)}{(t-2)/2}}= 12B \sqrt{\frac{\kappa \ln (2t)}{t-2}}.
    \end{align}
    Moreover, by $t/2\le \zeta(t)\le (t+1)/2$ and $\kappa\ge 1$, we have
    \begin{align}
        \delta(t)
        &=2\delta_0(t-\zeta(t),\epsilon_1(\zeta(t)))+\delta_1(\zeta(t))\\
        &=4 \exp \left( -\frac{2(t-\zeta(t))}{B^2}\cdot\frac{72B^2 \kappa\ln(4\zeta(t))}{\zeta(t)-1} \right)+\frac{1}{(2\zeta(t))^{\kappa}}\\
        &\le 4\exp \left( -\frac{t-1}{B^2}\cdot\frac{72B^2 \kappa\ln (4\cdot t/2)}{(t-1)/2} \right)+\frac{1}{(2\cdot t/2)^\kappa}
        = \frac{4}{(2t)^{144\kappa}}+\frac{1}{t^\kappa}
        \le \frac{2}{t^{\kappa}}.
    \end{align}
    Then, it follows that $h_t^*(\fD)\ge g_{\zeta(t)}(\fD)$ with probability at least $1-2/t^{\kappa}$. 
    By \Cref{thm:uniform-comp-and-regret} and \Cref{col:delta-epsilon}, the regret of $h_t^*$ for $(\fD,p;T)$ is bounded as follows:
    \begin{align}
        \regret 
        &\le 2t_0B + \sum_{t=2t_0+1}^T(6\epsilon(t)+B\delta(t))
        \le 2t_0B + \sum_{t=2t_0+1}^T \left( 72B \sqrt{\frac{\kappa \ln (2t)}{t-2}} + \frac{2B}{t^{\kappa}}\right) \\
        &\le 2t_0B + 72B\sqrt{\kappa\ln(2T)}\cdot\sum_{t=1}^{T-2} \frac{1}{\sqrt{t}} + 4B\cdot\sum_{t=1}^T\frac{1}{t}\\
        &\le 2t_0B + 72B\sqrt{\kappa\ln(2T)}\cdot\left(1+\int_{t=1}^{T-2} \frac{\mathrm{d}t}{\sqrt{t}}\right) + 4B\cdot\int_{t=1}^{T}\frac{\mathrm{d}t}{t}\\
        &=   2t_0B + 72B\sqrt{\kappa\ln(2T)}\cdot\left(1+\left[2\sqrt{T-2}-2\right]\right) + 4B\cdot\left[\ln(T)-\ln(1)\right]\\
        &= O(Bt_0 + B\sqrt{\kappa T\log T}). \qedhere
    \end{align}

\subsection{Proof of \Cref{thm:PI-refined-restate}}
\label{subsec:pf-thm:PI-refined}

Jin et al.~\cite{jin2024sample} established that the sample complexity of the prophet inequality is at most $\frac{25\ln^2\left(2e/\delta\right)}{\epsilon^2}$, which implies that, with $t-1$ samples from the distribution $\fD \in \Dadv$, we can construct an algorithm $h'_t[\History{t}]$ such that the assumption \eqref{eq:sample-complexity} holds with
\begin{align}
\label{eq:epsilon-delta-PI}
    \epsilon_1(t) \coloneqq \frac{5 \ln (4et)}{\sqrt{t-1}} \text{ and }
    \delta_1(t)\coloneqq \frac{1}{2t}.
\end{align}
Define $\zeta(t)\coloneqq \lfloor (t+1)/2 \rfloor$.
For $t \in [T]$, let $h^{\mathrm{PI}*}_t$ be the algorithm that uses $h'_{\zeta(t)}$ when $\cC_t$ occurs and $g^{\mathrm{PI}}_{\zeta(t)}$ otherwise, where $\epsilon_1$ and $\delta_1$ are given by~\eqref{eq:epsilon-delta-PI}.

    For $t \ge 2$ by \Cref{thm:uniform-comp-and-regret}, we have
    \begin{align}
    h_t^*(\fD)\ge g_{\zeta(t)}[\History{\zeta(t)}](\fD)
    \end{align}
    with probability at least $1-\delta(t)$.

    We now evaluate $\epsilon(t)$ and $\delta(t)$ for $t\ge 3$.
    Since $\zeta(t)\ge t/2$ and $\epsilon(t)$ is monotonically decreasing, we have
    \begin{align}
        \epsilon(t)
        =\epsilon_1(\zeta(t))
        \le \epsilon_1(t/2) = \frac{5\sqrt{2}  \ln (2et)}{\sqrt{t-2}}.
    \end{align}
    Additionally, since $t/2 \leq \zeta(t) \le (t+1)/2$, we have
    \begin{align}
        \delta(t)
        &=2\delta_0(t-\zeta(t),\epsilon_1(\zeta(t)))+\delta_1(\zeta(t))\\
        &\le 2\delta_0(t-\zeta(t),\epsilon_1((t+1)/2))+\delta_1(t/2)\\
        &=4 \exp \left( -2(t-\zeta(t))\cdot\frac{50\ln^2 (2e(t+1))}{t-1} \right)
        + \frac{1}{t}\\
        &\le 4\exp \left( -(t-1)\cdot\frac{50\ln^2 (2e(t+1))}{t-1} \right)
        + \frac{1}{t}\\
        &=\frac{4}{(2e(t+1))^{50\ln(2e(t+1))}}+\frac{1}{t}
        \le \frac{2}{t}.
    \end{align}
    Then $h_t^*(\fD)$ dominates $g_{\zeta(t)}(\fD)$ with probability at least $1-2/t$ for $t> \lfloor\sqrt{T}\rfloor$.
    Furthermore, we can obtain a regret bound $O(\sqrt{T}\log T)$ by a similar evaluation as in the proof of~\Cref{thm:main}.

\subsection{Proof of \Cref{thm:lower-bound}}
\label{subsec:pf-lower-bound}

To prove the theorem, we first define three distributions $D_0,D_+,D_-$, supported on $\{0,1/2,1\}$. Let $\epsilon \in (0,1/6)$ be a parameter to be chosen appropriately later. 
The distribution $D_0$ assigns $1/3$ to each value. The distribution $D_+$ assigns probability $1/3 + \epsilon$ to $1$, probability $1/3 - \epsilon$ to $1/2$, and probability $1/3$ to $0$. $D_-$ assigns probability $1/3 - \epsilon$ to $1$, probability $1/3 + \epsilon$ to $1/2$, and probability $1/3$ to $0$. 
Let $\fD_0=(D_0,D_0;\Piid)$, $\fD_+=(D_+,D_+;\Piid)$, and $\fD_-=(D_-,D_-;\Piid)$.

For the distributions $\fD_+$ and $\fD_-$, we have the following claim.
\begin{claim}
    \label{claim:per-round-regret}    
    Let $\fD \in \{\fD_+,\fD_-\}$, and $p \in \{\prwd,\pbest\}$. 
    Suppose that $h$ is a deterministic algorithm for the optimal stopping problem $(\fD, p)$ that does not achieve the optimal online profit. Then, we have
    \begin{align}
        \opton(\fD;p) - h(\fD;p) \ge \epsilon / 12.
    \end{align}
\end{claim}

\begin{proof}[Proof of~\Cref{claim:per-round-regret}]
For any distribution supported on $\{0,1/2,1\}$, it suffices to consider algorithms that (i) always accept the first variable realizing $1$ and reject the first variable realizing $0$, and (ii) always accept the second variable if the first variable is rejected.
Thus, with $2$ variables $X_1$ and $X_2$, there are only two deterministic algorithms: $h_+$ and $h_-$. Here, $h_+$ rejects $X_1$ when $X_1=1/2$, while $h_-$ accepts $X_1$ when $X_1=1/2$.

For the expected reward maximization problem (i.e., the profit function is $\prwd$), the expected profit of $h_+$ and $h_-$ under $\fD_+$ and $\fD_-$ can be calculated as follows:
\begin{align}
    h_+(\fD_+;\prwd) &= \Big( \frac13 + \epsilon \Big)\cdot 1 + \Big( \frac13 - \epsilon \Big)\cdot \frac{1+\epsilon}{2} + \frac13\cdot\frac{1+\epsilon}{2}, \\
    h_-(\fD_+;\prwd) &= \Big( \frac13 + \epsilon \Big)\cdot 1 + \Big( \frac13 - \epsilon \Big)\cdot \frac{1}{2} + \frac{1}{3}\cdot\frac{1+\epsilon}{2}, \\
    h_+(\fD_-;\prwd) &= \Big( \frac{1}{3} - \epsilon \Big)\cdot 1 + \Big( \frac{1}{3} + \epsilon \Big)\cdot \frac{1-\epsilon}{2} + \frac{1}{3}\cdot\frac{1-\epsilon}{2}, \\
    h_-(\fD_-;\prwd) &= \Big( \frac13 - \epsilon \Big)\cdot 1 + \Big( \frac13 + \epsilon \Big)\cdot \frac{1}{2} + \frac{1}{3}\cdot\frac{1-\epsilon}{2}.
\end{align}
Then, when the given distribution is $\fD_+$, $h_+$ achieves the optimal online profit and $h_-$ incurs a per-round regret of at least
\begin{align}
    h_+(\fD_+;\prwd) - h_-(\fD_+;\prwd) = \Big(\frac{1}{3} - \epsilon \Big) \cdot \frac{\epsilon}{2} \ge \frac{\epsilon}{12},
\end{align}
where the inequality is due to $0<\epsilon \le 1/6$. When the given distribution is $\fD_-$, $h_-$ achieves the optimal online profit and $h_+$ incurs a per-round regret of at least
\begin{align}
    h_-(\fD_-;\prwd) - h_+(\fD_-;\prwd) = \Big(\frac{1}{3} + \epsilon \Big) \cdot \frac{\epsilon}{2} \ge \frac{\epsilon}{6}.
\end{align}

Similarly, for the best choice probability maximization problem (i.e., the profit function is $\pbest$), the expected profit of $h_+$ and $h_-$ under $\fD_+$ and $\fD_-$ can be calculated as follows:
\begin{align}
    h_+(\fD_+;\pbest) &= \Big( \frac13 + \epsilon \Big)\cdot 1 + \Big( \frac13 - \epsilon \Big)\cdot \frac{2}{3} + \frac{1}{3}\cdot 1, \\
    h_-(\fD_+;\pbest) &= \Big( \frac13 + \epsilon \Big)\cdot 1 + \Big( \frac13 - \epsilon \Big)\cdot \Big(\frac{2}{3} - \epsilon \Big)+ \frac{1}{3}\cdot 1, \\
    h_+(\fD_-;\pbest) &= \Big( \frac{1}{3} - \epsilon \Big)\cdot 1 + \Big( \frac{1}{3} + \epsilon \Big)\cdot \frac{2}{3} + \frac{1}{3}\cdot 1, \\
    h_-(\fD_-;\pbest) &= \Big( \frac13 - \epsilon \Big)\cdot 1 + \Big( \frac13 + \epsilon \Big)\cdot \Big(\frac{2}{3} + \epsilon \Big) + \frac{1}{3}\cdot 1.
\end{align}
Then, when the given distribution is $\fD_+$, $h_+$ achieves the optimal online profit and $h_-$ incurs a per-round regret of at least
\begin{align}
    h_+(\fD_+;\pbest) - h_-(\fD_+;\pbest) = \Big(\frac{1}{3} - \epsilon \Big) \cdot \epsilon \ge \frac{\epsilon}{6},
\end{align}
where the inequality is due to $0<\epsilon \le 1/6$. When the given distribution is $\fD_-$, $h_-$ achieves the optimal online profit and $h_+$ incurs a per-round regret of at least
\begin{align}
    h_-(\fD_-;\pbest) - h_+(\fD_-;\pbest) = \Big(\frac{1}{3} + \epsilon \Big) \cdot \epsilon \ge \frac{\epsilon}{3}.
\end{align}
\end{proof}


    By Yao's principle, it suffices to show a regret lower bound for any deterministic selection rule against the uniform distribution over $\{\fD_+, \fD_-\}$.
    Fix an arbitrary deterministic selection rule, and let $h_t$ be a deterministic algorithm chosen in round $t$. Since we consider a deterministic rule, $h_t$ is determined by the previous samples $\History{t}=(\bX^{(1)},\ldots,\bX^{(t-1)})$ and does not have any other randomness, i.e., $h_t$ is either $h_+$ or $h_-$ for each round $t \in [T]$.

    For any $T$ samples $\mathbf{S}_T = (\bX^{(1)},\ldots,\bX^{(T)}) \in \{0,1/2,1\}^{2 \times T}$,
    let $f_+(\mathbf{S}_T)$ (resp. $f_-(\mathbf{S}_T)$) denote the number of $t$ such that $h_t = h_+$ (resp. $h_t = h_-$) when $\mathbf{S}_T$ is realized.
    Let $\Pr_+[\cdot]$, $\Pr_-[\cdot]$, and $\Pr_0[\cdot]$ denote the probability when the underlying distribution is $\fD_+$, $\fD_-$, and $\fD_0$, respectively. Similarly, $\E_+[\cdot]$, $\E_-[\cdot]$, and $\E_0[\cdot]$ denote the expectation when the underlying distribution is $\fD_+$, $\fD_-$, and $\fD_0$, respectively.
    We use the following claim.
    \begin{claim}
        \label{claim:lower-bound}
        The following inequalities hold:
        \begin{align}
            \E_+[f_+(\mathbf{S}_T)] &\le \E_0[f_+(\mathbf{S}_T)] + 2\epsilon T \sqrt{T}, \label{eq:claim-plus} \\
            \E_-[f_-(\mathbf{S}_T)] &\le \E_0[f_-(\mathbf{S}_T)] + 2\epsilon T \sqrt{T}. \label{eq:claim-minus}
        \end{align}
    \end{claim}
    \begin{proof}[Proof of~\Cref{claim:lower-bound}]
        We first prove \eqref{eq:claim-plus}. It follows that
        \begin{align}
            \E_+[f_+(\mathbf{S}_T)] - \E_0[f_+(\mathbf{S}_T)]
            & = \sum_{\mathbf{S}_T \in \{0,1/2,1\}^{2 \times T}} f_+(\mathbf{S}_T) \Big( \Pr_+[\mathbf{S}_T] - \Pr_0[\mathbf{S}_T] \Big) \\
            & \le \sum_{\substack{\mathbf{S}_T \in \{0,1/2,1\}^{2 \times T} \\ \mathbf{S}_T: \Pr_+[\mathbf{S}_T] \ge \Pr_0[\mathbf{S}_T]}} f_+(\mathbf{S}_T) \Big( \Pr_+[\mathbf{S}_T] - \Pr_0[\mathbf{S}_T] \Big) \\
            & \le T \cdot \sum_{\substack{\mathbf{S}_T \in \{0,1/2,1\}^{2 \times T} \\ \mathbf{S}_T: \Pr_+[\mathbf{S}_T] \ge \Pr_0[\mathbf{S}_T]}} \Big( \Pr_+[\mathbf{S}_T] - \Pr_0[\mathbf{S}_T] \Big) \\
            & = \frac{T}{2} \cdot \sum_{\mathbf{S}_T \in \{0,1/2,1\}^{2 \times T}} \big| \Pr_+[\mathbf{S}_T] - \Pr_0[\mathbf{S}_T] \big| \\
            & \le \frac{T}{2} \cdot \sqrt{2 \KLdiv{\Pr_0}{\Pr_+}},
        \end{align}
        where the last inequality is due to Pinsker's inequality and $\KLdiv{\Pr_0}{\Pr_+}$ denotes the KL-divergence between $\Pr_0$ and $\Pr_+$, i.e.,
        \begin{align}
            \KLdiv{\Pr_0}{\Pr_+} \coloneqq \sum_{\mathbf{S}_T \in \{0,1/2,1\}^{2 \times T}} \Pr_0[\mathbf{S}_T] \ln \frac{\Pr_0[\mathbf{S}_T]}{\Pr_+[\mathbf{S}_T]}.
        \end{align}
        Since $\mathbf{S}_T = ((X_1^{(1)},X_2^{(1)}),\ldots,(X_1^{(T)},X_2^{(T)}))$ follows product distribution $D_0^{2T}$ or $D_+^{2T}$, we can evaluate $\KLdiv{\Pr_0}{\Pr_+}$ as follows:
        \begin{align}
            \KLdiv{\Pr_0}{\Pr_+}
            & = 2T \cdot \KLdiv{D_0}{D_+} \\
            & = 2T \left( \frac13 \ln \frac{1/3}{1/3+\epsilon} + \frac13 \ln \frac{1/3}{1/3-\epsilon} + \frac13 \ln \frac{1/3}{1/3} \right) \\
            & = - \frac{2T}{3} \ln(1-9\epsilon^2) \\
            & \le \frac{2T}{3} \cdot 4 \ln(4/3) \cdot 9\epsilon^2,
        \end{align}
        where the last inequality is due to the convexity of $-\ln(1-x)$ and $9\epsilon^2 \le 1/4$. Then, by letting $C=2\sqrt{3\ln(4/3)}\approx 1.86 <2$, we have \eqref{eq:claim-plus} as follows:
        \begin{align}
            \E_+[f_+(\mathbf{S}_T)] - \E_0[f_+(\mathbf{S}_T)]
            \le \frac{T}{2} \sqrt{2 \cdot \frac{2T}{3} \cdot 4 \ln(4/3) \cdot 9\epsilon^2}
            = C\epsilon T \sqrt{T} < 2\epsilon T \sqrt{T}.
        \end{align}

        The proof of \eqref{eq:claim-minus} is analogous.
    \end{proof}
    We return to the proof of Theorem~\ref{thm:lower-bound}.
    By~\Cref{claim:lower-bound}, it holds that
    \begin{align}
        \frac12 (\E_+[f_+(\mathbf{S}_T)] + \E_-[f_-(\mathbf{S}_T)])
        & \le \frac12 \cdot \E_0[f_+(\mathbf{S}_T) + f_-(\mathbf{S}_T)] + 2\epsilon T \sqrt{T} 
        = \frac{T}{2} + 2\epsilon T \sqrt{T}.
    \end{align}
    This implies that, against the uniform distribution over $\{\fD_+,\fD_-\}$, any deterministic rule can select the ``correct algorithm'' ($h_+$ for distribution $\fD_+$ or $h_-$ for distribution $\fD_-$) only $T/2 + 2\epsilon T \sqrt{T}$ times in expectation. By~\Cref{claim:per-round-regret}, any deterministic selection rule incurs cumulative regret at least
    \begin{align}
        \frac{\epsilon}{12} \cdot \Big( T - \big(T/2 + 2\epsilon T \sqrt{T}\big) \Big)
        = \frac{\epsilon}{12} \cdot \big(T/2 - 2\epsilon T \sqrt{T}\big).
    \end{align}
Tuning $\epsilon = 1/(8\sqrt{T})$ yields the lower bound $\Omega(\sqrt{T})$. 

\section{Linear Regret from Repeated Use of a Competitive Algorithm}
\label{sec:example}

In this section, we show an example in which simply adopting a single-sample algorithm in every round leads to a linear regret.

Let $\epsilon$ be a positive real that is less than $1/2$.
Consider the repeated prophet inequality with two random variables and $T$ rounds where:
\begin{itemize}
    \item $X_1 = 1/2$ with probability $1$,
    \item $X_2 = 1$ with probability $1/2+\epsilon$, $X_2=0$ with probability $1/2-\epsilon$.
\end{itemize}
For this instance, the single-sample online algorithm~\cite{rubinstein_et_al:LIPIcs:2020:11745} for the prophet inequality operates as follows:
\begin{enumerate}
    \item If the first sample is $(1/2,1)$ which occurs with probability $1/2+\epsilon)$, then the single-sample algorithm chooses the first value $\geq 1$, and hence the expected reward in each round is $1/2+\epsilon$.
    \item If the first sample is $(1/2,0)$, then the single-sample algorithm always chooses the value $1/2$ of $X_1$.
\end{enumerate}
We assume that, for the first round, it accepts $X_2$.
Then, the expected total reward of the single-sample online algorithm over $T$ rounds is 
\begin{align}
\MoveEqLeft
(1/2+\epsilon)\cdot (1+(T-1)(1/2+\epsilon)) + (1/2-\epsilon)\cdot (0+(T-1)\cdot1/2) \\
&= T\cdot (1/2+\epsilon) - (T-1)\cdot (\epsilon/2-\epsilon^2)
\end{align}

On the other hand, the optimal online algorithm always chooses the value of $X_2$ because $\mathbb{E}[X_2]=1/2+\epsilon > 1/2=\mathbb{E}[X_1]$.
Then the expected total reward of the optimal online algorithm is $T\cdot (1/2+\epsilon)$.

Therefore, the regret is $(T-1)(\epsilon/2-\epsilon^2)$, which increases linearly with respect to $T$.

\end{document}